\documentclass[a4paper,12pt]{amsart}
\usepackage[usenames]{color}

\usepackage{fullpage}
\usepackage{qtree}
\usepackage[dvips]{graphicx}
\usepackage{gastex}
\usepackage{amsmath}
\usepackage{amssymb}
\usepackage{amsfonts}
\usepackage{amscd}
\usepackage{amsbsy}

\usepackage{verbatim}
\usepackage[T1]{fontenc}
\usepackage[latin1]{inputenc}
\usepackage{enumerate}

\definecolor{Gray}{rgb}{0.9,0.9,0.9}
\definecolor{Gray2}{rgb}{0.6,0.6,0.6}

\usepackage{latexsym}
\usepackage{subfigure}

\theoremstyle{plain}
\newtheorem{theorem}{Theorem}[section]

 \newtheorem{lemma}[theorem]{Lemma}
 \newtheorem{proposition}[theorem]{Proposition}
 
 \newtheorem{corollary}[theorem]{Corollary}
\theoremstyle{definition}
 \newtheorem{definition}[theorem]{Definition}

\theoremstyle{remark}
 \newtheorem{remark}[theorem]{Remark}
 \newtheorem{example}[theorem]{Example}
\theoremstyle{remark}

\numberwithin{equation}{section}

\newcommand{\A}{\mathcal{A}}
\newcommand{\B}{\mathcal{B}}
\newcommand{\F}{\mathcal{F}}
\newcommand{\G}{\mathcal{G}}
\newcommand{\N}{\mathcal{N}}
\newcommand{\T}{\mathcal{T}}

\newcommand{\TT}{\mathfrak{T}}

\DeclareMathOperator{\Sh}{\mathsf{X}}

\DeclareMathOperator{\im}{Im}

\begin{document}
\title[Cellular automata between sofic tree shifts]{Cellular automata between sofic tree shifts}
\author[T. Ceccherini-Silberstein]{Tullio Ceccherini-Silberstein}
\address{Dipartimento di Ingegneria, Universit\`a del Sannio, C.so
Garibaldi 107, 82100 Benevento, Italy}
\email{tceccher@mat.uniroma3.it}
\author[M. Coornaert]{Michel Coornaert}
\address{Institut de Recherche Math\'ematique Avanc\'ee,
UMR 7501, Universit\'e  de Strasbourg et CNRS,
7 rue Ren\'e-Descartes,67000 Strasbourg, France}
\email{coornaert@math.unistra.fr}
\author[F. Fiorenzi]{Francesca Fiorenzi}
\address{Laboratoire de Recherche en Informatique,
Universit\'e Paris-Sud 11,
91405 Orsay Cedex, France}
\email{fiorenzi@lri.fr}
\author[Z. \v Suni\'c]{Zoran \v Suni\'c}
\address{Department of Mathematics, Texas A\&M University, MS-3368, College Station, TX 77843-3368, USA}
\email{sunic@math.tamu.edu}
\date{\today}
\begin{abstract}
We study the sofic tree shifts of $A^{\Sigma^*}$, where $\Sigma^*$ is a regular rooted tree of finite rank.
In particular, we give their characterization in terms of unrestricted Rabin automata.
We show that if $X \subset A^{\Sigma^*}$ is a sofic tree shift, then the
configurations in $X$ whose orbit under the shift action is finite are
dense in $X$, and, as a consequence of this, we deduce that every
injective cellular automata $\tau\colon X \to X$ is surjective.
Moreover, a characterization of sofic tree shifts in terms of general Rabin automata is given.

We present an algorithm for establishing whether two unrestricted Rabin automata accept the same sofic tree shift or not.
This allows us to prove the decidability of the surjectivity problem for
cellular automata between sofic tree shifts.
We also prove the decidability of the injectivity problem for cellular automata defined on a tree shift of finite type.
\end{abstract}

\maketitle

\section{Introduction}
Cellular automata have been studied in several different settings and from various points of view.
Classically, the universe is the grid $\mathbb{Z}^d$ of integer points of the Euclidean $d$-dimensional space. The state of every cell in the grid ranges in a finite alphabet $A$ and a configuration is an element of $A^{\mathbb{Z}^d}$, that is, a map $f \colon \mathbb{Z}^d \to A$ that describes the state of every cell. A cellular automaton (CA) is a map $\tau \colon A^{\mathbb{Z}^d} \to A^{\mathbb{Z}^d}$ that changes a configuration by simultaneously updating the state of each cell according to a fixed local rule, i.e., a rule that only considers the states of the neighbors of this cell.
Later, the notion of a CA has been extended to the case where the universe is the Cayley graph of a finitely generated group or semigroup $G$ or, even
more generally, a locally finite graph admitting a dense holonomy in
the sense of Gromov (see~\cite{Gromov99} and \cite{livre}). Note that the grid $\mathbb{Z}^d$ is the Cayley graph of the free abelian group of rank $d$. A subshift $X \subset A^G$ is a set of configurations avoiding a given set of forbidden patterns.
In~\cite{Fiorenzi03}, \cite{Fiorenzi04} and~\cite{Fiorenzi09}, the case of CAs defined on the subshifts of $A^G$, where $G$ is a finitely generated group, has been studied (see also \cite{CeccheriniCoornaert12}, \cite{CeccheriniCoornaert12b}, \cite{CeccheriniCoornaert12c} and~\cite{CeccheriniCoornaert13}).

In this paper, we study cellular automata between subshifts of $A^{\Sigma^*}$ (also called tree shifts),
where $\Sigma$ is a finite set and $\Sigma^*$ is the free
monoid (therefore of finite rank) generated by $\Sigma$.
The Cayley graph of $\Sigma^*$ is a regular $\vert\Sigma\vert$-ary rooted tree.
We investigate, in particular, the decidability of the surjectivity and injectivity problems for these CAs.
For $\vert\Sigma\vert=1$, this corresponds to the case of one-sided CAs (i.e., of CAs defined on $A^\mathbb{N}$).
In~\cite{FiciFiorenzi12}, some topological properties of CAs defined on $A^{\Sigma^*}$ has been investigated.

\smallskip

In the classical setting, it is well known that most of the decidability problems concerning
CAs turn out to have a positive answer in the
one-dimensional case ($d = 1$, including the one-sided case). For instance, the surjectivity and the
injectivity problems have been proved decidable by Amoroso and
Patt~\cite{AmorosoPatt72}. On the other hand,
Kari~\cite{Kari90,Kari94} proved that these problems fail to
be decidable in dimension $d\geq2$. We also mention that Kari~\cite{kari:nilpotent} proved that already in dimension one it is undecidable whether a given CA is nilpotent or not, that is, whether its limit set contains just one configuration or not.
Recently, Margenstern~\cite{Margenstern09} proved the undecidability
of the injectivity problem for certain CAs defined on the hyperbolic plane.

Muller and Schupp~\cite{MullerSchupp81,MullerSchupp85} introduced and
studied the class of context-free graphs and proved that the monadic second-order logic (MSOL) of any such graph is decidable. Since surjectivity and injectivity of
a given CA are both expressible in the language of MSOL,
they deduced that these properties are decidable in the case of
CAs with universe a context-free graph. Note that this result covers the one-dimensional
cases and, more generally, the case of free groups.
Moreover, since the Cayley graph of a finitely generated free monoid
is context-free, the result of Muller and Schupp also applies in the
case of CAs defined on $A^{\Sigma^*}$, where $\Sigma^*$ is a regular rooted tree.
The main purpose of this paper is to explicitly describe two decision algorithms
for the surjectivity and the injectivity problems of these CAs. Moreover, our algorithms apply to CAs defined on suitable classes of tree shifts of $A^{\Sigma^*}$.

\smallskip

Tree shifts of $A^{\Sigma^*}$ have been extensively studied by Aubrun and Béal in~\cite{Aubrun11}, \cite{AubrunBeal10} and \cite{AubrunBeal12}.
In the present work we use a slightly different (but equivalent) setting. Some results in Section~\ref{s:Unrestricted Rabin automata} have already been proved by Aubrun and Béal and naturally extend the corresponding one-dimensional properties.
We give a direct proof of them for the sake of completeness and in order to fix notation.

A tree shift $X \subset A^{\Sigma^*}$ is said to be of finite type if it can be described as the set of configurations avoiding a finite number of forbidden patterns.
A sofic tree
shift is defined as the image of a tree shift of finite type under a CA.
In the one-dimensional case, a sofic subshift of $A^\mathbb{N}$ may be
characterized as the set of all right-infinite words over $A$ accepted by some
finite-state automaton. In our setting, we use the notion of an unrestricted Rabin
automaton (see \cite{Rabin69}, \cite{ThatcherWright68}, \cite{tata07}), as well as a related notion of acceptance, in order to provide
the analogous characterization of sofic tree shifts of $A^{\Sigma^*}$. If $\A$ is an unrestricted Rabin automaton, we denote by $\Sh_\A$ the (sofic) tree shift accepted by $\A$.

The unrestricted Rabin automata we consider are called top-down: assuming a graphical representation of trees with the root symbol at the top, an automaton starts its computation at
the root and moves downward.
%% NEW
%Now, the same way one associates with a one-dimensional subshift $X \subset A^\mathbb{N}$ (resp. a
%finite automaton accepting a sofic subshift $X \subset A^\mathbb{N}$) a language $L(X) \subset A^*$ (resp. a finite automaton with initial and final states accepting the associated language $L(X)$), in
%our setting, we associate with every tree shift $X \subset A^{\Sigma^*}$ a ``language'' $\TT(X)$ of
%\emph{full-tree-patterns} and we consider a suitably defined notion of finite-tree automaton (this
%is just an unrestricted Rabin automaton with specified initial and final states) together with
%a corresponding notion of acceptance for full-tree-patterns. This way, recognizable tree languages are the languages recognized by finite-tree-automata.
%%
Many authors consider bottom-up unrestricted Rabin automata instead.
Actually,
top-down unrestricted Rabin automata and bottom-up unrestricted Rabin automata have the same expressive power. The bottom-up version of the characterization of sofic tree shifts in terms of acceptance has been proved by Aubrun and B\'eal~\cite{AubrunBeal12}.

Unrestricted Rabin automata constitute a device that allows us to prove (Theorem~\ref{t:regular}) that the regular configurations (i.e., with a finite orbit under the action of $\Sigma^*$) of a sofic tree shift $X \subset A^{\Sigma^*}$ are dense in $X$ (with respect to the prodiscrete topology).
This is no more the case when the free monoid $\Sigma^*$ is replaced by a free group.
The density of regular configurations allows us to prove (Corollary~\ref{c:surjunctive}) the surjunctivity of a CA defined on a sofic tree shift (a selfmap is said to be surjunctive if the implication ``injective $\Longrightarrow$ surjective'' holds). This is also the case for CAs on residually finite groups and on irreducible sofic subshifts of $A^\mathbb{Z}$ (see~\cite{Fiorenzi09}, \cite{livre} and~\cite{CeccheriniCoornaert12}).

For a subshift $X \subset A^{\mathbb{N}}$, the notion of irreducibility expresses the fact that for any two connected patterns of $X$ there always exists a configuration in which they appear simultaneously as factors. If $X$ is sofic, this notion is equivalent to the strong connectedness of a suitable finite-state automaton accepting $X$. Lind and Marcus~\cite{LindMarcus95} prove this fact for the sofic subshifts of $\mathbb{Z}$. In our setting, we introduce a notion of irreducibility for tree shifts of $A^{\Sigma^*}$ and a notion of strong connectedness
for unrestricted Rabin automata. We then show that the mentioned equivalence still holds for sofic tree shifts.
Aubrun and Béal~\cite{AubrunBeal10} prove an analogous result in which both irreducibility and strong connectedness are stronger properties than the ones we use.

\smallskip
Rabin automata are more general than the ones we consider, because they have initial states and accepting conditions. This is the reason why we call ``unrestricted'' our automata accepting sofic tree shifts (each state is initial and there are no restrictions on the accepting sets).
It is shown by Rabin~\cite{Rabin69} that the class of tree languages (i.e., subsets of $A^{\Sigma^*}$) accepted by Rabin automata is closed under complement, and that it is decidable if a language accepted by a Rabin automaton is empty. This implies for example that it is decidable whether a tree shift accepted by an unrestricted Rabin automaton is the full tree shift $A^{\Sigma^*}$ or not (one first constructs the Rabin automaton which accepts the complement and then checks for emptiness). Nevertheless, even if the emptiness algorithm for Rabin automata is not difficult (even though the problem is NP-complete), the complementation of Rabin automata is a highly complex issue, and is fairly impractical.
In this paper we have worked out the details of the complementation and emptiness algorithms in a limited setting of interest, namely when we start with an unrestricted Rabin automaton.

Let us now illustrate our decidability results. It is easy to decide whether the sofic tree shift accepted by a given unrestricted Rabin automaton is empty or not. An idea to decide the surjectivity of a cellular automaton $\tau \colon A^{\Sigma^*} \to A^{\Sigma^*}$ could be to establish the emptiness of $A^{\Sigma^*} \setminus \tau(A^{\Sigma^*})$.
But given a nontrivial tree shift $X \subset A^{\Sigma^*}$, the complement
$A^{\Sigma^*} \setminus X$ always fails to be a tree shift. To avoid this obstruction, we use a pattern set (whose elements are called full-tree-patterns), which is accepted by suitably defined finite-tree automaton (i.e., an unrestricted Rabin automaton in which initial and final states are specified in order to accept, or recognize, finite configurations).
%A full-tree-pattern is a map $p \colon T \to A$ defined on a finite full subtree $T \subset \Sigma^*$ (i.e., a finite subtree in which every vertex other than the leaves has the same number of children).
For $X \subset A^{\Sigma^*}$, we denote by $\TT(X)$ the set of the full-tree-patterns of $X$.
For one-dimensional subshifts $X\subset A^\mathbb{N}$, the set $\TT(X)$ coincides with the language of $X$. In any case, a tree shift is entirely determined by its full-tree-patterns.
With this notion at hand we are able to construct a finite-tree automaton accepting exactly the full-tree-patterns of a sofic tree shift, as stated in the following result.

\begin{theorem}\label{t:subsetCONST2}
Let $\A$ be an unrestricted Rabin automaton. Then there is an effective procedure to construct a co-determini\-stic finite-tree automaton $\A_{\rm cod}(\mathcal{I}, F)$ such that
$\TT(\Sh_\A) = \TT(\A_{\rm cod}(\mathcal{I}, F)).$
\end{theorem}

An important difference between bottom-up and
top-down unrestricted Rabin automata appears in the question of determinism since deterministic
top-down unrestricted Rabin automata are strictly less powerful than nondeterministic ones
and therefore are strictly less powerful than bottom-up unrestricted Rabin automata.
Intuitively,
this is due to the fact that tree properties specified by deterministic
top-down unrestricted Rabin automata can depend only on path properties.
This is the reason why we deal with co-determinism: this notion is now equivalent to the determinism in the bottom-up setting and
deterministic and nondeterministic bottom-up unrestricted Rabin automata have the same expressive
power.
The nonequivalence between deterministic and nondeterministic top-down unrestricted Rabin automata is proved in~\cite[Proposition 1.6.2]{tata07}
for the general setting of finite ordered trees
with bounded rank. We give in Example~\ref{e:countnondet} another proof of this fact for our particular class of unrestricted Rabin automata.

The recognizable sets of full-tree-patterns form a class which is closed under complementation (Theorem~\ref{t:complement}). This allows to prove
that the complement of the full-tree-pattern set of a sofic tree shift is a recognizable set of full-tree-patterns.

\begin{theorem}\label{t:complement2}
Let $\A$ be an unrestricted Rabin automaton. Then there is an effective procedure to construct a co-complete and co-deterministic finite-tree automaton $\A_\complement(I,F)$ (with a single initial state) which accepts the complement of the full-tree-pattern set of $\Sh_\A$, in formual\ae, $\TT(A^{\Sigma^*}) \setminus \TT(\Sh_\A) = \TT(\A_\complement(I,F))$.
\end{theorem}

The following result is mentioned in Rabin's paper~\cite{Rabin69} and addressed by Doner~\cite{Doner70} and Thatcher and Wright~\cite{ThatcherWright68} in a more general setting.
\begin{theorem}[Emptiness problem]\label{t:emptiness}
Let $\A(\mathcal{I},F)$ be a finite-tree automaton. Then there is an algorithm which establishes whether $\TT(\A(\mathcal{I},F)) = \varnothing$ or not.
\end{theorem}

The three effective procedures above are the first step towards the solution of the following decision problem.

\begin{theorem}[Equality problem]\label{t:equality}
It is decidable whether two unrestricted Rabin automata accept the same tree shift or not.
\end{theorem}

Note that a particular case of this result can also be deduced for strongly connected unrestricted Rabin automata (in the sense of Aubrun and B\'eal), by a minimization process. Actually, in~\cite{AubrunBeal10} it is shown that there exists a canonical minimal presentation of a ``strongly'' (always in the sense of Aubrun and B\'eal) irreducible sofic tree shift. Thus the decision algorithm consists in computing the minimal automata of the two shifts and checking whether they coincide or not.

The decidability results by Amoroso and Patt for CAs defined on $A^\mathbb{Z}$ have been generalized in \cite{Fiorenzi09} to CAs defined on the subshifts of finite type of $A^\mathbb{Z}$. These results use the notion of a de Bruijn graph.
This idea can be extended to our setting to get an unrestricted Rabin automaton accepting the image (under a CA) of a tree shift of finite type of $A^{\Sigma^*}$. The graph underlying this unrestricted Rabin automaton is a sort of ``multidimensional'' de Bruijn graph. Our final decision results, which use both this construction and Theorem~\ref{t:equality}, yield two algorithms for establishing the surjectivity and the injectivity of CAs. Moreover, the following solutions for the surjectivity
(resp. injectivity) problem for CAs defined on sofic tree shifts (resp.
on sofic tree shifts of finite type) hold.

\begin{theorem}[Surjectivity problem]\label{t:surjectivity} Let $A$ and $B$ be two finite alphabets.
Also let $X \subset A^{\Sigma^*}$ and $Y \subset B^{\Sigma^*}$ be two sofic tree shifts.
Then it is decidable whether a cellular automaton $\tau \colon X \to Y$ is surjective or not.
\end{theorem}

\begin{theorem}[Injectivity problem]\label{t:injectivity} Let $A$ and $B$ be two finite alphabets.
Also let $X \subset A^{\Sigma^*}$ be a tree shift of finite type.
Then it is decidable whether a cellular automaton $\tau \colon X \to B^{\Sigma^*}$ is injective or not.
\end{theorem}

We conclude our paper by recalling the original notion of a Rabin automaton in its full generality (see~\cite{Rabin69}).
In particular, as stated in the following theorem, we give another description of the class of sofic tree shifts in terms of general Rabin automata.

\begin{theorem}\label{t:rabin+shift=sofic}
A tree shift $X \subset A^{\Sigma^*}$ is sofic if and only if it is recognized by a Rabin automaton. In other words, the intersection of the class of Rabin recognizable tree languages and the class of tree shifts is precisely the class of sofic tree shifts.
\end{theorem}

Part of this paper was presented at
the conference CIAA 2012 (see~\cite{CeccheriniCoornaertFiorenziSunic12}).

\subsection*{Acknowledgements}
We express our deep gratitude to the unknown referee for
      his/her most careful reading of our paper and for several
      valuable remarks and suggestions.

\section{Definitions and background material}

\subsection{The rooted regular tree $\Sigma^*$}\label{ss:free}
Let $\Sigma$ be a nonempty finite set.
%For a positive integer $k$, we denote by $\Sigma$ the set $\{0, 1, \dots, k-1\}$.
Given a nonnegative integer $n \in \mathbb{N}$ we denote by $\Sigma^n$ the set of all \emph{words} $w = \sigma_1\sigma_2 \cdots \sigma_n$ of \emph{length} $n$ (where $\sigma_i \in \Sigma$ for $i=1,2,\dots,n$) over $\Sigma$. In particular, $\varepsilon \in \Sigma^0$ indicates the only word of length $0$, called the \emph{empty word}. For $n \geq 1$, we denote by $\Delta_n$ the set $\bigcup_{i = 0}^{n-1}\Sigma^i$ (that is, the set of all words of length $\leq n-1$). In particular, we have $\Delta_1 = \{\varepsilon\}$ and $\Delta_2 = \{\varepsilon\} \cup \Sigma$.

The \emph{concatenation} of two words $w = \sigma_1\sigma_2 \cdots \sigma_n \in \Sigma^n$ and $w' = \sigma'_1\sigma'_2 \cdots \sigma'_m\in \Sigma^m$ is the word $ww' = \sigma_1\sigma_2 \cdots \sigma_n\sigma'_1\sigma'_2 \cdots \sigma'_m\in \Sigma^{m+n}$. Then the set $\Sigma^* = \bigcup_{n \in \mathbb{N}} \Sigma^n$, equipped with the multiplication given by concatenation, is a monoid with identity element the empty word $\varepsilon$. It is called the \emph{free monoid} over the set $\Sigma$.

From the graph theoretical point of view, we consider $\Sigma^*$ as the vertex set of the regular $k$-ary rooted tree, where $k = \vert \Sigma\vert$. The empty word $\varepsilon$ is its \emph{root} and, for every vertex $w \in \Sigma^*$, the vertices $w\sigma\in \Sigma^*$ (with $\sigma \in \Sigma$) are called the \emph{children} of $w$. Every vertex is connected to each of its children by an %non-labeled
edge.

\begin{remark}\label{r:factor closed}
Notice that the subsets $\Delta_n$ are \emph{factor closed}, that is, $ww' \in \Delta_n$ implies $w, w' \in \Delta_n$ for all $w, w' \in \Sigma^*$.
\end{remark}

\subsection{Configurations and shift spaces}
Let $A$ be a finite set, called the \emph{alphabet}. The elements of $A$ are called \emph{letters} or \emph{colors}. The \emph{space of configurations} of $\Sigma^*$ over the alphabet $A$ is the set $A^{\Sigma^*}$ of
all maps $f \colon  \Sigma^* \to  A$. When equipped with the \emph{prodiscrete topology} (that is, with the product topology  where each factor $A$ of $A^{\Sigma^*} = \prod_{w \in \Sigma^*}A$ is endowed with the discrete topology), the configuration space is a compact, totally disconnected, metrizable space.
Also, the free monoid $\Sigma^*$ has a right action on $A^{\Sigma^*}$ defined as follows: for every $w\in \Sigma^*$ and $f \in  A^{\Sigma^*}$ the configuration $f^w \in  A^{\Sigma^*}$ is defined by setting
$$f^w(w') = f(ww')$$ for all $w' \in \Sigma^*$. This action, called the \emph{shift action}, is continuous with respect to the prodiscrete topology.

\medskip

Recall that a sub-basis for the prodiscrete topology on $A^{\Sigma^*}$ consists of the \emph{elementary cylinders}
$$\mathcal{C}(w,a) = \{f \in A^{\Sigma^*} : f(w) = a\},$$
where $w \in \Sigma^*$ and $a \in A$.
Moreover, a neighborhood basis of a configuration $f \in A^{\Sigma^*}$ is given by the sets
$$\N(f,n) = \{g \in A^{\Sigma^*} : g\vert_{\Delta_n} = f\vert_{\Delta_n}\}$$
where $n \geq 1$ (as usual, for $M\subset \Sigma^*$, we denote by $f\vert_M$ the restriction of $f$ to $M$).

\begin{definition}[Tree shift]
A subset $X \subset A^{\Sigma^*}$ is called a \emph{tree shift} (or \emph{subshift}, or simply \emph{shift}) provided that $X$ is closed (with
respect to the prodiscrete topology) and \emph{shift-invariant} (that is, $f^w \in X$ for all
$f \in X$ and $w \in \Sigma^*$). In particular $A^{\Sigma^*}$ is a tree shift and it is called the \emph{full (tree) $A$-shift}.
\end{definition}

\subsection{Forbidden blocks and shifts of finite type}
Let $\Sigma$ be a nonempty finite set and let $A$ be a finite alphabet.

\begin{definition}[Pattern and block]
Let $M\subset \Sigma^*$ be a finite set. A \emph{pattern} is a map $p \colon  M \to  A$. The set $M$ is called the \emph{support} of $p$ and it is denoted by ${\rm supp}(p)$. We denote by $A^M$ the set of all patterns with support $M$.
For any $n\geq1$, a \emph{block} is a pattern $p \colon  \Delta_n \to  A$. The integer $n$ is called the \emph{size} of the block.
The set of all blocks is denoted by $\B(A^{\Sigma^*})$.
\end{definition}

If $X$ is a subset of $A^{\Sigma^*}$ and $M \subset \Sigma^*$ is finite,
the set of patterns $\{f\vert_M : f \in X\}$ is denoted by $X_M$. For $n \geq 1$, the notation $X_n$ is an abbreviation for $X_{\Delta_n}$ (that is, the set of all blocks of size $n$ which are restrictions to $\Delta_n$ of some configuration in $X$). We denote by $\B(X)$ the set of all blocks of $X$ (that is, $\B(X)= \bigcup_{n\geq1}X_n$).

Given a block $p \in \B(A^{\Sigma^*})$ and a configuration $f \in A^{\Sigma^*}$, we say that $p$ \emph{appears} in $f$ if there exists $w \in \Sigma^*$ such that $(f^w)\vert_{{\rm supp}(p)} = p$. If $p$ does not appear in $f$, we say that $f$ \emph{avoids} $p$.
Let $\F$ be a set of blocks. We denote by $\mathsf{X}(\F)$ the set of all configurations in $A^{\Sigma^*}$ avoiding simultaneously all the blocks in $\F$, in symbols
$$\mathsf{X}(\F) = \{f \in A^{\Sigma^*} : (f^w)\vert_{\Delta_n} \notin \F, \textup{ for all } w \in \Sigma^* \textup{ and } n \geq 1 \}.$$

If $\vert \Sigma \vert = 1$ we have a one-dimensional setting in which $\Sigma^*$ is identified with $\mathbb{N}$. Indeed, if $\Sigma = \{\sigma\}$, we associate $n\in \mathbb{N}$ with $\sigma^n\in \Sigma^n$, where $\sigma^n$ denotes the word $\underbrace{\sigma\sigma\cdots \sigma}_n$. In this case, a configuration $f \in A^\mathbb{N}$ can be identified with the (right) infinite word $w = a_0a_1\cdots$ over the alphabet $A$ where $a_0 = f(\varepsilon)$ and $a_n = f(\sigma^n)$ for all $n\geq 1$. Analogously, a block of size $n$ can be identified with an element of $A^n$, that is a word of length $n$ over the alphabet $A$. Indeed the set $\Delta_n = \{\varepsilon, \sigma, \sigma\sigma, \dots, \sigma^{n-1}\} \subset \Sigma^*$ is identified with $\{0, 1, 2, \dots, n-1\} \subset \mathbb{N}$.

By analogy with the one-dimensional case (see for example \cite[Theorem 6.1.21]{LindMarcus95}), we have the following combinatorial characterization of tree shifts.

\begin{proposition}\label{p:forbidden}
A subset $X  \subset A^{\Sigma^*}$ is a tree shift if and only if there exists a set
$\F \subset \B(A^{\Sigma^*})$ of blocks such that
$X = \mathsf{X}(\F)$.
\end{proposition}

Let $X \subset A^{\Sigma^*}$ be a tree shift. A set $\F$ of blocks as in Proposition~\ref{p:forbidden} is called a \emph{defining set of forbidden blocks} for $X$.
If one can find a finite defining set of forbidden blocks for $X$ one says that $X$ is a tree shift of \emph{finite type}.

\begin{remark}\label{r:B(X) determines X}The blocks of a tree shift determine the tree shift. In fact, given two tree shifts $X,Y \subset A^{\Sigma^*}$, we have
$X = \Sh(\B(A^{\Sigma^*}) \setminus \B(X))$
so that
$X = Y \Longleftrightarrow \B(X) = \B(Y).$ Moreover, $\B(A^{\Sigma^*})\setminus \B(X)$ is the largest defining set of forbidden blocks for $X$ (with respect to inclusion).
\end{remark}

\begin{remark}\label{r:memory}
Let $\F = \{p_1, \dots p_h\}$ be a finite set of blocks with supports $\Delta_{n_1}, \dots, \Delta_{n_h}$, respectively. Let $n = \max\{n_1, \dots, n_h\}$. Consider the set $\F' \subset A^{\Delta_n}$ consisting of all blocks $p$ such that $p\vert_{\Delta_{n_i}} = p_i$ for some $i \in \{1, \dots, h\}$. In other words, $\F'$ consists of all the possible extensions to $\Delta_n$ of any block in $\F$.
It is clear that $\F'$ is finite and $\mathsf{X}(\F) = \mathsf{X}(\F')$.
For this reason, we can always suppose that the forbidden blocks of a defining set of a given tree shift of finite type all have the same support. This motivates the following definition:
\end{remark}

\begin{definition}[Memory]\label{d:memory}
A tree shift of finite type has \emph{memory} $n$ if it admits a defining set of forbidden blocks all of size $n$.
\end{definition}
Notice that if a tree shift has memory $n$, then it also has memory $m$ for all $m\geq n$.

\subsection{Examples of shifts}
Let $\Sigma$ be a nonempty finite set and let $A$ be a finite alphabet.
%If $\vert \Sigma \vert = k$ we may identify $\Sigma$ with the set  $\{0, 1, \dots, k-1\}$.
\begin{example}[Full tree shift]
The space $A^{\Sigma^*}$ of all configurations over $A$ is clearly a tree shift of finite type. Indeed, the empty set is the unique defining set of forbidden blocks for $A^{\Sigma^*}$.
\end{example}

\begin{example}[Monochromatic children]\label{e:monochromaticchildren}
Consider the set of blocks
$$\F = \left\{
{\scriptsize
\begin{tabular}[c]{@{}l@{}}
%\Tree [.$a$ $...$ !\qsetw{0.05cm} $a_\sigma$ !\qsetw{0.05cm} $...$ !\qsetw{0.05cm} $...$ !\qsetw{0.05cm} $a_{\sigma'}$ !\qsetw{0.05cm} $...$ ]
\Tree [.$a$ $...\phantom{a}$ !\qsetw{0.05cm} $a_\sigma$ !\qsetw{0.05cm} $...$ !\qsetw{0.05cm} $\phantom{a}a_{\sigma'}$ !\qsetw{0.05cm} $\phantom{a}...$ ]
%\Tree [.$a$ $...$ $a_\sigma$ $...$ ]
\end{tabular}
}
\in A^{\Delta_2}  : a_\sigma \neq a_{\sigma'} \textup{ for some } \sigma, \sigma' \in \Sigma\right\}.$$
The tree shift $\Sh(\F) \subset A^{\Sigma^*}$ is of finite type and exactly consists of those configurations for which every vertex in $\Sigma^*$ has monochromatic children.
If $\vert \Sigma \vert = 2$ and $A = \{a,b\}$ an example of a configuration in $\Sh(\F)$ is sketched in Figure~\ref{ALBmonochromaticchildren}.
\end{example}
\begin{figure}[!h]
%%%%%%%%%%%%%%%%%%%%%%%%%%%%%%%%%%%%%%%%%%%%%%%%%%%%%%%%%%%%%%%%%%%%%%%
\qtreecenterfalse
\scriptsize
\Tree [.$a$ [.$b$ [.$a$ [.$a$ $...$ !\qsetw{0.1cm} $...$ ] !\qsetw{0.5cm} [.$a$ $...$ !\qsetw{0.1cm} $...$ ] ] !\qsetw{1.2cm} [.$a$ [.$b$ $...$ !\qsetw{0.1cm} $...$ ] !\qsetw{0.5cm} [.$b$ $...$ !\qsetw{0.1cm} $...$ ] ] ] !\qsetw{3cm} [.$b$ [.$b$ [.$b$ $...$ !\qsetw{0.1cm} $...$ ] !\qsetw{0.5cm} [.$b$ $...$ !\qsetw{0.1cm} $...$ ] ] !\qsetw{1.2cm} [.$b$ [.$b$ $...$ !\qsetw{0.1cm} $...$ ] !\qsetw{0.5cm} [.$b$ $...$ !\qsetw{0.1cm} $...$ ] ] ] ]
%%%%%%%%%%%%%%%%%%%%%%%%%%%%%%%%%%%%%%%%%%%%%%%%%%%%%%%%%%%%%%%%%%%%%%%
\caption{A configuration of the tree shift in Example~\ref{e:monochromaticchildren}.}\label{ALBmonochromaticchildren}
\end{figure}

\begin{example}[Even sum]\label{e:zerosummodtwo}
Let $A = \{0,1\}$.
Consider the set of blocks
$$\F = \left\{
{\scriptsize
\begin{tabular}[c]{@{}l@{}}
\Tree [.$a$ $...\phantom{a}$ !\qsetw{0.05cm} $a_\sigma$ !\qsetw{0.05cm} $...$ !\qsetw{0.05cm} $\phantom{a}a_{\sigma'}$ !\qsetw{0.05cm} $\phantom{a}...$ ]
\end{tabular}
}
\in A^{\Delta_2}  : a + \sum_{\sigma \in \Sigma} a_\sigma \equiv 1 \mod 2\right\}.$$
The tree shift $\Sh(\F) \subset A^{\Sigma^*}$ is of finite type and exactly consists of those configurations such that the sum of the label of any vertex in $\Sigma^*$ with the labels of its children is always even.
If $\vert \Sigma \vert = 2$, an example of configuration of $\Sh(\F)$ is given in Figure~\ref{ALBzerosummod2}.
\end{example}
\begin{figure}[!h]
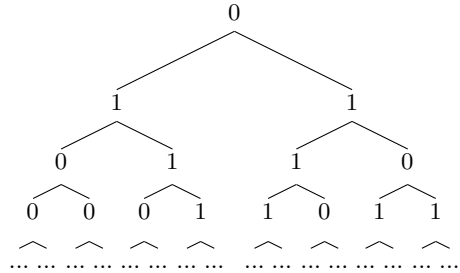

%%%%%%%%%%%%%%%%%%%%%%%%%%%%%%%%%%%%%%%%%%%%%%%%%%%%%%%%%%%%%%%%%%%%%%%
\qtreecenterfalse
\scriptsize
\Tree [.$0$ [.$1$ [.$0$ [.$0$ $...$ !\qsetw{0.1cm} $...$ ] !\qsetw{0.5cm} [.$0$ $...$ !\qsetw{0.1cm} $...$ ] ] !\qsetw{1.2cm} [.$1$ [.$0$ $...$ !\qsetw{0.1cm} $...$ ] !\qsetw{0.5cm} [.$1$ $...$ !\qsetw{0.1cm} $...$ ] ] ] !\qsetw{3cm} [.$1$ [.$1$ [.$1$ $...$ !\qsetw{0.1cm} $...$ ] !\qsetw{0.5cm} [.$0$ $...$ !\qsetw{0.1cm} $...$ ] ] !\qsetw{1.2cm} [.$0$ [.$1$ $...$ !\qsetw{0.1cm} $...$ ] !\qsetw{0.5cm} [.$1$ $...$ !\qsetw{0.1cm} $...$ ] ] ] ]
%%%%%%%%%%%%%%%%%%%%%%%%%%%%%%%%%%%%%%%%%%%%%%%%%%%%%%%%%%%%%%%%%%%%%%%
\caption{A configuration of the tree shift in Example~\ref{e:zerosummodtwo}.}\label{ALBzerosummod2}
\end{figure}

\subsection{Cellular automata and sofic tree shifts}
Let $\Sigma$ be a nonempty finite set. Let $A$ and $B$ two finite alphabets, and let $X \subset A^{\Sigma^*}$ be a tree shift.
\begin{definition}[Cellular automaton]
A map $\tau \colon X \to B^{\Sigma^*}$ is called a \emph{cellular automaton} (CA for short) if there exist a finite subset $M \subset \Sigma^*$ and a map $\mu \colon  A^M \to  B$ such that
$$\tau(f)(w) = \mu((f^w)\vert_M)$$
for all $f \in X$ and $w \in \Sigma^*$. The set $M$ is called a \emph{memory set} for $\tau$ and $\mu$ is the associated \emph{local defining map}.
\end{definition}

\begin{remark}\label{r:samealphabet}
%In the definition above we have assumed that the alphabet of the shift $X \subset A^{\Sigma^*}$ is the same as the alphabet of its image $\tau(X)$. In this assumption there is no loss of generality because if $\tau \colon X \to B^{\Sigma^*}$, one can always consider $X$ and $B^{\Sigma^*}$ as two subshifts of $(A \cup B)^{\Sigma^*}$.
Classically, a cellular automaton is often a selfmapping $\tau \colon X \to X$ (in particular
$A = B$). By dropping this hypothesis, we deal with a more general notion that, in the one-dimensional case, corresponds to that of \emph{sliding block code} as defined in~\cite{LindMarcus95}.
\end{remark}

\begin{remark}\label{r:triangle}
Let $\tau \colon   X \to  B^{\Sigma^*}$ be a CA with memory set $M$ and local defining map $\mu \colon  A^M \to  B$.
If $M' \subset \Sigma^*$ is a finite subset containing $M$, consider the map $\mu' \colon  A^{M'} \to  B$ defined by $\mu'(p) = \mu(p\vert_M)$, for all $p \in A^{M'}$.
Then $M'$ and $\mu'$ are respectively a memory set and a local defining map for $\tau$ as well.
In view of this fact, we assume in the sequel (without loss of generality), that a memory set has the form $M = \Delta_n$, for a suitable $n\geq 1$.
\end{remark}

The following is a topological characterization of cellular automata. For a proof in the one-dimensional case, see \cite[Theorem 6.2.9]{LindMarcus95}. See also \cite[Theorem 1.8.1]{livre} and \cite[Proposition 1.2.4]{Fiorenzi00thesis}, for a more general setting.

\begin{theorem}[Curtis-Hedlund-Lyndon]\label{t:curtis}
A map $\tau \colon   X \to  B^{\Sigma^*}$ is a CA if and only if it is continuous (with respect to the prodiscrete topology) and commutes with the shift action (that is, $(\tau(f))^w = \tau(f^w)$ for all $f \in X$ and $w \in \Sigma^*$).
\end{theorem}

\begin{remark}\label{r:image}
Given a cellular automaton $\tau \colon X \to  B^{\Sigma^*}$, it immediately follows from Theorem~\ref{t:curtis} and the compactness of $X$, that the image $\tau(X) \subset B^{\Sigma^*}$ is again a tree shift.
\end{remark}

\begin{definition}[Sofic tree shift]\label{d:sofic}
A tree shift $X \subset B^{\Sigma^*}$ is called \emph{sofic} provided there exist a finite
alphabet $A$, a tree shift of finite type $Y \subset A^{\Sigma^*}$, and a CA
$\tau \colon Y \to B^{\Sigma^*}$ such that $X = \tau(Y)$.
\end{definition}

%\begin{definition}{\rm[Sofic shift]\label{soficDEF}
%A subshift $X \subset A^{\Sigma^*}$ is called \emph{sofic} provided there
%exist a finite alphabet $B$, a subshift of finite type $Y \subset B^{\Sigma^*}$, and a CA $\tau \colon   Y \to  A^{\Sigma^*}$
%such that $X = \tau(Y)$.
%}\end{definition}

\begin{remark}
Every tree shift of finite type is sofic.
\end{remark}

\begin{example}[A sofic tree shift which is not of finite type]\label{e:evenshift}
Let $A = \{0,1\}$ and let $\vert \Sigma \vert = 1$. The subshift $Y = \Sh(\F)\subset A^\mathbb{N}$ where $\F = \{11\}$, is called the \emph{golden mean shift}. Thus $Y$ is the subshift of finite type containing all the right infinite words avoiding the factor $11$. Consider the cellular automaton $\tau \colon  Y \to  A^\mathbb{N}$ defined by the local defining map $\mu \colon  A^2 \to  A$ associating $0$ with $01$ and $10$, and associating $1$ with $00$ and $11$. As one can see, the image $\tau(Y)$ is precisely the subshift $X \subset \{0,1\}^\mathbb{N}$ consisting of all right infinite words in which, between two occurrences of $1$s, there always is an even number of occurrences of $0$s. More precisely, $X = \Sh(\F')\subset A^\mathbb{N}$ where $\F' = \{10^{2n+1}1 : n \in \mathbb{N}\}$.
For this reason, $X$ is called the \emph{even shift}. Thought sofic (by definition), it is easy to see that $X$ cannot be of finite type (see \cite[Example 2.1.5, Example 2.1.9]{LindMarcus95}). Indeed, if $N \in \mathbb{N}$ was the memory of $X$, the block $10^{2N+1}1$ would be forbidden with a length greater than~$N$.

The notion of even shift can be generalized to any $\Sigma$ as follows.
%Let $\vert \Sigma \vert = 1$ and $A = \{0,1\}$.
%The subshift $X = \Sh(\F)\subset A^\mathbb{N}$ where $\F = \{11\}$, is called the \emph{golden mean shift}. Thus $X$ is the set of the right infinite words avoiding the word $11$. Consider the CA $\tau \colon  X \to A^\mathbb{N}$ defined by the local map $\mu \colon  A^2 \to  A$ associating $0$ with $01$ and $10$, and associating $1$ with $00$ and $11$. As one can see, the image $\tau(X)$ is the sofic subshift $Y$ in which between two occurrences of $1$, there are always an even number of occurrences of $0$.
%In symbols, $Y = \Sh(\F')\subset A^\mathbb{N}$ where $\F' = \{10^{2n+1}1 : n \in \mathbb{N}\}$. For this reason, $Y$ is called the \emph{even shift}. This shift is not of finite type as proved in \cite[Example 2.1.5, Example 2.1.9]{LindMarcus95}.
%
%Obviously the definition of even shift can be given in general for any $\vert \Sigma \vert \geq 1$. In this case the words of $\F'$ are forbidden on any branch of the tree $\Sigma^*$. %This shift is still sofic as proved in Example~\ref{e:michel}.
%Given two words $u, v \in \Sigma^*$ there exists a unique
%geodesic path connecting $u$ to $uv$, say $\pi = (w_0, w_1, \dots, w_n)$ with
%$w_0 = u$ and $w_n = uv$.
We define $X \subset A^{\Sigma^*}$ as the tree shift
consisting of all configurations $f \colon \Sigma^* \to A$ avoiding the words of $\F'$ on any branch of the tree $\Sigma^*$.
Again, by using Remark~\ref{r:memory} and Definition~\ref{d:memory}, it can be seen that $X$ is not of finite type.
In Example~\ref{e:even1dimensional} we shall deduce from Proposition~\ref{p:XA sofic} and
Corollary~\ref{c:sofic iff accepted} that $X$ is sofic.

\end{example}

\section{Unrestricted Rabin graphs and automata}\label{s:Unrestricted Rabin automata}
\begin{definition}
An \emph{unrestricted Rabin graph} (over $\Sigma$ and with alphabet set $A$), is a $4$-tuple $\G = (S,\Sigma,A,\T)$, where
\begin{itemize}
\item $\Sigma$ is a nonempty finite set and $A$ is a finite alphabet;
\item $S$ is a nonempty set, called the set of \emph{states} (or \emph{vertices}) of $\G$;
\item $\T$ is a subset of $S \times A \times S^\Sigma$ whose elements are called \emph{transition bundles}.
\end{itemize}
When the state set $S$ is finite, an unrestricted Rabin graph $\A = (S,\Sigma,A,\T)$ is called an \emph{unrestricted Rabin automaton}.

Given a transition bundle $t = (s;a;(s_\sigma)_{\sigma \in \Sigma}) \in \T$ we denote by
$\mathbf{i}(t) := s \in S$ its \emph{source state}, by $\lambda(t) := a \in A$ its \emph{label}, by $\mathbf{t}(t) := (s_\sigma)_{\sigma \in \Sigma} \in S^\Sigma$ its \emph{terminal sequence} and by $\mathbf{t}_\sigma(t) :=  s_\sigma \in S$ its $\sigma$-\emph{terminal state}.
A \emph{bundle loop on $s \in S$} is a transition bundle $t \in \T$ such that $\mathbf{i}(t) = \mathbf{t}_\sigma(t) = s$ for all $\sigma \in \Sigma$.

An unrestricted Rabin graph $\G = (S,\Sigma,A,\T)$ is said to be \emph{essential} provided that each state $s \in S$ is the source state of some transition bundle.
\end{definition}

\begin{remark}
 Our notion of essentiality differs with the one used in the two-sided one-dimensional case (see~\cite[Definition 2.2.9]{LindMarcus95}). In this setting we want that  each state is the source of some transition bundle but we allow states which are not terminal. Our trees being rooted, we do not need to go backward. This fails to hold in the above mentioned
two-sided one-dimensional case where the automaton has to recognize biinfinite words.
\end{remark}

Note that \emph{multiple transition bundles}, that is, different transition bundles having the same source state and the same terminal sequence, are allowed.
Nevertheless, the number of transition bundles between a given source state and a given terminal sequence must be finite (cannot exceed $\vert A \vert$): two such bundles must either have different labels or coincide if they have the same label.

\subsection{Graphical representation}
Let $\Sigma$ be a nonempty finite set. If $\vert \Sigma \vert = k$ we may identify $\Sigma$ with the set  $\{0, 1, \dots, k-1\}$.
This way, a transition bundle of an unrestricted Rabin automaton $\A = (S,\Sigma,A,\T)$ is a $(k+2)$-tuple $t = (s;a;s_0, \dots, s_{k-1})$ which can be visualized as in Figure~\ref{FIGautomaton}.
If $\vert \Sigma \vert = 2$ and $(s;a;s_0, s_1)$ is a transition bundle, we represent the edge from $s$ to $s_0$ by a broken line and the edge from $s$ to $s_1$ by a full line. This makes unnecessary the need to pay attention to the graphical position of the $i$-th child of the state~$s$ (see Figure~\ref{FIG2bundle}).
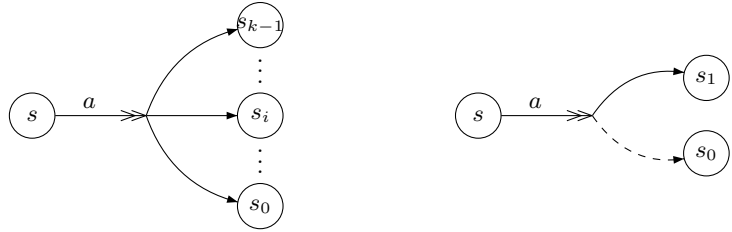
\begin{figure}[h!]
\scriptsize
\centering
\subfigure[A general labeled transition bundle.]{
\begin{picture}(50,30)(-25,-18)
\gasset{Nw=6,Nh=6}
\node(S)(-15,0){$s$}
\node(SK)(15,12){$s_{k-1}$}
\node[Nframe=n,Nw=-1,Nh=-1](SD)(15,7){$\vdots$}
\node(SI)(15,0){$s_i$}
\node[Nframe=n,Nw=-1,Nh=-1](SDD)(15,-5){$\vdots$}
\node(S1)(15,-12){$s_0$}
\node[Nframe=n,Nw=-1,Nh=-1](SS)(0,0){}
\drawedge[AHnb=2,AHangle=20,AHLength=2,AHlength=0](S,SS){$a$}
\drawedge[curvedepth=3](SS,SK){}
\drawedge(SS,SI){}
\drawedge[curvedepth=-3](SS,S1){}
\end{picture}\label{FIGautomaton}
}\quad\quad
%%%%%%%%%%%%%%%%%%%%%%%%%%%%%%%%%%%%%%%%%%%%%%%%%%%%%%%%%%%%%%%%%%%%%%%
\subfigure[A labeled transition bundle of an unrestricted Rabin automaton in which $\Sigma = \{0,1\}$.]{
\begin{picture}(50,16)(-25,-18)
\gasset{Nw=6,Nh=6}
\node(S)(-15,0){$s$}
\node(S1)(15,5){$s_1$}
\node(S2)(15,-5){$s_0$}
\node[Nframe=n,Nw=-1,Nh=-1](SS)(0,0){}
\drawedge[AHnb=2,AHangle=20,AHLength=2,AHlength=0](S,SS){$a$}
\drawedge[curvedepth=3](SS,S1){}
\drawedge[curvedepth=-3,dash={1}0](SS,S2){}
\end{picture}\label{FIG2bundle}
}
\caption{Representations of a transition bundle.}
\label{FIGbundles}
\end{figure}

\subsection{Acceptance}
Let $\Sigma$ be a nonempty finite set and let $A$ be a finite alphabet.
\begin{definition}[Unrestricted Rabin graph of a configuration]
\label{d:configuration}
The \emph{unrestricted Rabin graph of a configuration $f \in A^{\Sigma^*}$} is defined by $\G_f = (\Sigma^*,\Sigma,A,\T_f)$
where
$$\T_f = \{(w;f(w);(w\sigma)_{\sigma \in \Sigma}): w \in \Sigma^*\}.$$
\end{definition}

\begin{example}
Let $\Sigma = \{0,1\}$ and $A = \{a,b\}$. We denote by $f \in A^{\Sigma^*}$ the configuration sketched in Figure~\ref{ALBmonochromaticchildren}. The unrestricted Rabin graph of $f$ is illustrated in Figure~\ref{FIGrabingraph}.

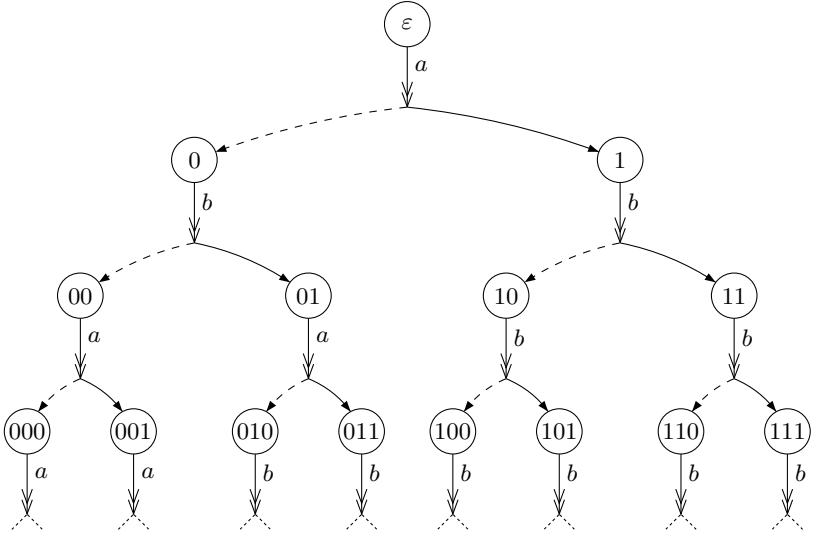
\begin{figure}[h!]
\scriptsize

%%%%%%%%%%%%%%%%%%%%%%%%%%%%%%%%%%%%%%%%%%%%%%%%%%%%%%%%%%%%%%%%%%%%%%%
\begin{center}
\begin{picture}(0,71)(0,-13)
\gasset{Nw=6,Nh=6}
\node(E)(0,58){$\varepsilon$}
\node[Nframe=n,Nw=-1,Nh=-1](EE)(0,47){}

\node(0)(-28,40){$0$}
\node[Nframe=n,Nw=-1,Nh=-1](0S)(-28,29){}
\node(1)(28,40){$1$}
\node[Nframe=n,Nw=-1,Nh=-1](1S)(28,29){}

\node(00)(-43,22){$00$}
\node[Nframe=n,Nw=-1,Nh=-1](00S)(-43,11){}
\node(01)(-13,22){$01$}
\node[Nframe=n,Nw=-1,Nh=-1](01S)(-13,11){}
\node(10)(13,22){$10$}
\node[Nframe=n,Nw=-1,Nh=-1](10S)(13,11){}
\node(11)(43,22){$11$}
\node[Nframe=n,Nw=-1,Nh=-1](11S)(43,11){}

\node(000)(-50,4){$000$}
\node[Nframe=n,Nw=-1,Nh=-1](000S)(-50,-7){}
\node(001)(-36,4){$001$}
\node[Nframe=n,Nw=-1,Nh=-1](001S)(-36,-7){}
\node(010)(-20,4){$010$}
\node[Nframe=n,Nw=-1,Nh=-1](010S)(-20,-7){}
\node(011)(-6,4){$011$}
\node[Nframe=n,Nw=-1,Nh=-1](011S)(-6,-7){}
\node(100)(6,4){$100$}
\node[Nframe=n,Nw=-1,Nh=-1](100S)(6,-7){}
\node(101)(20,4){$101$}
\node[Nframe=n,Nw=-1,Nh=-1](101S)(20,-7){}
\node(110)(36,4){$110$}
\node[Nframe=n,Nw=-1,Nh=-1](110S)(36,-7){}
\node(111)(50,4){$111$}
\node[Nframe=n,Nw=-1,Nh=-1](111S)(50,-7){}
\node[Nframe=n,Nadjust=w,Nh=5](0000)(-53,-10){$$}
\node[Nframe=n,Nadjust=w,Nh=5](0001)(-47,-10){$$}
\node[Nframe=n,Nadjust=w,Nh=5](0010)(-39,-10){$$}
\node[Nframe=n,Nadjust=w,Nh=5](0011)(-33,-10){$$}
\node[Nframe=n,Nadjust=w,Nh=5](0100)(-23,-10){$$}
\node[Nframe=n,Nadjust=w,Nh=5](0101)(-17,-10){$$}
\node[Nframe=n,Nadjust=w,Nh=5](0110)(-9,-10){$$}
\node[Nframe=n,Nadjust=w,Nh=5](0111)(-3,-10){$$}
\node[Nframe=n,Nadjust=w,Nh=5](1000)(3,-10){$$}
\node[Nframe=n,Nadjust=w,Nh=5](1001)(9,-10){$$}
\node[Nframe=n,Nadjust=w,Nh=5](1010)(17,-10){$$}
\node[Nframe=n,Nadjust=w,Nh=5](1011)(23,-10){$$}
\node[Nframe=n,Nadjust=w,Nh=5](1100)(33,-10){$$}
\node[Nframe=n,Nadjust=w,Nh=5](1101)(39,-10){$$}
\node[Nframe=n,Nadjust=w,Nh=5](1110)(47,-10){$$}
\node[Nframe=n,Nadjust=w,Nh=5](1111)(53,-10){$$}
\drawedge[AHnb=2,AHangle=20,AHLength=2,AHlength=0](E,EE){$a$}
\drawedge[AHnb=2,AHangle=20,AHLength=2,AHlength=0](0,0S){$b$}
\drawedge[AHnb=2,AHangle=20,AHLength=2,AHlength=0](1,1S){$b$}
\drawedge[AHnb=2,AHangle=20,AHLength=2,AHlength=0](00,00S){$a$}
\drawedge[AHnb=2,AHangle=20,AHLength=2,AHlength=0](01,01S){$a$}
\drawedge[AHnb=2,AHangle=20,AHLength=2,AHlength=0](10,10S){$b$}
\drawedge[AHnb=2,AHangle=20,AHLength=2,AHlength=0](11,11S){$b$}
\drawedge[AHnb=2,AHangle=20,AHLength=2,AHlength=0](000,000S){$a$}
\drawedge[AHnb=2,AHangle=20,AHLength=2,AHlength=0](001,001S){$a$}
\drawedge[AHnb=2,AHangle=20,AHLength=2,AHlength=0](010,010S){$b$}
\drawedge[AHnb=2,AHangle=20,AHLength=2,AHlength=0](011,011S){$b$}
\drawedge[AHnb=2,AHangle=20,AHLength=2,AHlength=0](100,100S){$b$}
\drawedge[AHnb=2,AHangle=20,AHLength=2,AHlength=0](101,101S){$b$}
\drawedge[AHnb=2,AHangle=20,AHLength=2,AHlength=0](110,110S){$b$}
\drawedge[AHnb=2,AHangle=20,AHLength=2,AHlength=0](111,111S){$b$}

{
\gasset{curvedepth=1}
\drawedge(EE,1){}
\drawedge(0S,01){}
\drawedge(1S,11){}
\drawedge(00S,001){$$}
\drawedge(01S,011){$$}
\drawedge(10S,101){$$}
\drawedge(11S,111){$$}
}
{
\gasset{curvedepth=-1,dash={1}0}
\drawedge(EE,0){}
\drawedge(0S,00){$$}
\drawedge(1S,10){$$}
\drawedge(00S,000){$$}
\drawedge(10S,100){$$}
\drawedge(01S,010){$$}
\drawedge(11S,110){$$}
}

\drawedge[AHnb=0,dash={0.4}0](000S,0000){}
\drawedge[AHnb=0,dash={0.4}0](000S,0001){}
\drawedge[AHnb=0,dash={0.4}0](001S,0010){}
\drawedge[AHnb=0,dash={0.4}0](001S,0011){}
\drawedge[AHnb=0,dash={0.4}0](010S,0100){}
\drawedge[AHnb=0,dash={0.4}0](010S,0101){}
\drawedge[AHnb=0,dash={0.4}0](011S,0110){}
\drawedge[AHnb=0,dash={0.4}0](011S,0111){}
\drawedge[AHnb=0,dash={0.4}0](100S,1000){}
\drawedge[AHnb=0,dash={0.4}0](100S,1001){}
\drawedge[AHnb=0,dash={0.4}0](101S,1010){}
\drawedge[AHnb=0,dash={0.4}0](101S,1011){}
\drawedge[AHnb=0,dash={0.4}0](110S,1100){}
\drawedge[AHnb=0,dash={0.4}0](110S,1101){}
\drawedge[AHnb=0,dash={0.4}0](111S,1110){}
\drawedge[AHnb=0,dash={0.4}0](111S,1111){}
\end{picture}
\end{center}\caption{The unrestricted Rabin graph of the configuration in Figure~\ref{ALBmonochromaticchildren}.}\label{FIGrabingraph}
\end{figure}

%%%%%%%%%%%%%%%%%%%%%%%%%%%%%%%%%%%%%%%%%%%%%%%%%%%%%%%%%%%%%%%%%%%%%%%

\end{example}

\begin{definition}[Homomorphism]\label{d:homomorphism}
Let $\G_1 = (S_1,\Sigma,A,\T_1)$ and $\G_2 = (S_2,\Sigma,A,\T_2)$ be two unrestricted Rabin graphs. A \emph{homomorphism} from $\G_1$ to $\G_2$
is a map $\alpha \colon  S_1 \to  S_2$ such that
\begin{equation}
(\alpha(s);a;(\alpha(s_\sigma))_{\sigma\in\Sigma}) \in \T_2
\end{equation}
for each $(s;a;(s_\sigma)_{\sigma\in\Sigma}) \in \T_1$.
By abuse of language/notation, we also denote by $\alpha \colon  \G_1 \to  \G_2$ such a
homomorphism.
\end{definition}

\begin{definition}[Acceptance]
Let $\A = (S,\Sigma,A,\T)$ be an unrestricted Rabin automaton. We say that a configuration $f \in A^{\Sigma^*}$ is \emph{accepted} (or \emph{recognized}) by $\A$, if there exists a homomorphism $\alpha \colon  \G_f \to  \A$, where $\G_f$ denotes the unrestricted Rabin graph of the configuration $f$.
In this case, we say that \emph{$f$ is accepted by $\A$ via} $\alpha$.\\
We denote by $\Sh_\A$ the set consisting of all those configurations $f \in A^{\Sigma^*}$ accepted by $\A$, that is
$$
\Sh_\A = \{f \in A^{\Sigma^*} : \textup{there exists a homomorphism } \alpha \colon  \G_f \to  \A\}.
$$
An unrestricted Rabin automaton $\A$ is called a \emph{presentation} for $X \subset A^{\Sigma^*}$ provided that $X = \Sh_\A$.
\end{definition}

\begin{remark}\label{r:essential}
We could always consider essential unrestricted Rabin automata. This is not restrictive since, by recursively removing all states that are source of no transition bundles together with all transition bundles admitting these states as terminal states, we can transform any unrestricted Rabin automaton $\A$ into an essential one $\A'$ which accepts the same configuration set, i.e., such that $\Sh_\A = \Sh_{\A'}$.
\end{remark}

\begin{definition}
A \emph{bundle automaton} is an unrestricted Rabin automaton $\mathfrak{A} = (S,\Sigma,A,\T)$ in which the labeling map $\lambda \colon \T \to A$ is injective. That is, two distinct elements in $\T$ carry different labels.
\end{definition}

\begin{definition}
Let $\A = (S,\Sigma,A,\T)$ be an unrestricted Rabin automaton. The \emph{underlying bundle automaton} of $\A$ is the bundle automaton $\mathfrak{A} = (S,\Sigma,\T,\T')$ defined over the alphabet $\T$, such that
$$\T' = \{(s;t;(s_\sigma)_{\sigma\in\Sigma}) : t \in \T \textup{ and } t = (s;a;(s_\sigma)_{\sigma\in\Sigma})\}.$$
\end{definition}
Roughly speaking, a transition bundle of $\mathfrak{A}$ is obtained by replacing the label of any transition bundle $t = (s;a;(s_\sigma)_{\sigma \in \Sigma})$ of $\A$
by $t$ itself.
Equivalently, a bundle $(s;t;(s_\sigma)_{\sigma\in\Sigma}) \in S \times \T \times S^\Sigma$ belongs to $\T'$ if and only if $\mathbf{i}(t) = s$ and $\mathbf{t}_\sigma(t) = s_\sigma$ for each $\sigma\in\Sigma$.

\begin{remark}\label{r:acceptance} Explicitly, a configuration $f \in A^{\Sigma^*}$ is accepted by an unrestricted Rabin automaton $\A = (S,\Sigma,A,\T)$ if there exists a map $\alpha \colon  \Sigma^* \to  S$ such that
\begin{equation}\label{eq:acceptance}
(\alpha(w);f(w);(\alpha(w\sigma))_{\sigma \in \Sigma}) \in \T\end{equation}
for all $w \in \Sigma^*$.

In particular, if $\mathfrak{A} = (S,\Sigma,\T,\T')$ is the bundle automaton underlying $\A$, a configuration $f \in \T^{\Sigma^*}$ is accepted by $\mathfrak{A}$ if and only if $\mathbf{t}_\sigma(f(w)) = \mathbf{i}(f(w\sigma))$,
for all $w \in \Sigma^*$ and $\sigma \in \Sigma$.
Indeed if $\mathfrak{A}$ accepts $f \in \T^{\Sigma^*}$ via $\alpha \colon  \Sigma^* \to  S$, then
$(\alpha(w);f(w);(\alpha(w\sigma))_{\sigma \in \Sigma}) \in \T'$
for all $w \in \Sigma^*$. This implies that $\mathbf{i}(f(w)) = \alpha(w)$ and $\mathbf{t}_\sigma(f(w)) = \alpha(w\sigma)$ for each $w \in \Sigma^*$ and $\sigma\in\Sigma$. Therefore $\mathbf{t}_\sigma(f(w)) = \alpha(w\sigma) = \mathbf{i}(f(w\sigma))$. Conversely, assume that $\mathbf{t}_\sigma(f(w)) = \mathbf{i}(f(w\sigma))$,
for all $w \in \Sigma^*$ and $\sigma \in \Sigma$. Then $\mathfrak{A}$ accepts $f$ via the homomorphism $\alpha \colon  \Sigma^* \to  S$ defined by $\alpha(w) = \mathbf{i}(f(w))$.

\end{remark}
As it will be proved in Proposition~\ref{p:XA sofic}, if $\A$ is an unrestricted Rabin automaton then the set $\Sh_\A \subset A^{\Sigma^*}$ is a sofic tree shift. Conversely, every sofic tree shift has the form $\Sh_\A$ for some unrestricted Rabin automaton $\A$ (see Corollary~\ref{c:sofic iff accepted}).

\begin{example}[Full shift]
The full $A$-shift $A^{\Sigma^*}$ is accepted by the bundle automaton $\A = (\{s\},\Sigma,A, \T)$, where $\T$ consists of $\vert A\vert$ bundle loops on $s$. That is, $\T = \{(s;a;(s)_{\sigma \in \Sigma}) : a \in A\}$.
If $\vert \Sigma \vert = 2$ and $A = \{a_1, a_2, a_3, a_4\}$, the corresponding automaton is represented in Figure~\ref{FIGfullshift}.
\end{example}
\begin{figure}[h!]
%%%%%%%%%%%%%%%%%%%%%%%%%%%%%%%%%%%%%%%%%%%%%%%%%%%%%%%%%%%%%%%%%%%%%%%
\scriptsize
\begin{center}
\begin{picture}(0,30)(0,-15)
\gasset{Nw=6,Nh=6}
\node(s)(0,0){$s$}
\node[Nframe=n,Nw=-1,Nh=-1](s-)(-15,0){}
\node[Nframe=n,Nw=-1,Nh=-1](s^)(0,15){}
\node[Nframe=n,Nw=-1,Nh=-1](s+)(15,0){}
\node[Nframe=n,Nw=-1,Nh=-1](s_)(0,-15){}

\drawedge[AHnb=2,AHangle=20,AHLength=2,AHlength=0,ELside=r](s,s-){$a_1$}
\drawedge[curvedepth=5,dash={1}0](s-,s){}
\drawedge[curvedepth=8](s-,s){}
\drawedge[AHnb=2,AHangle=20,AHLength=2,AHlength=0,ELside=r](s,s^){$a_2$}
\drawedge[curvedepth=5,dash={1}0](s^,s){}
\drawedge[curvedepth=8](s^,s){}
\drawedge[AHnb=2,AHangle=20,AHLength=2,AHlength=0,ELside=r](s,s+){$a_3$}
\drawedge[curvedepth=5,dash={1}0](s+,s){}
\drawedge[curvedepth=8](s+,s){}
\drawedge[AHnb=2,AHangle=20,AHLength=2,AHlength=0,ELside=r](s,s_){$a_4$}
\drawedge[curvedepth=5,dash={1}0](s_,s){}
\drawedge[curvedepth=8](s_,s){}
\end{picture}
\end{center}
%%%%%%%%%%%%%%%%%%%%%%%%%%%%%%%%%%%%%%%%%%%%%%%%%%%%%%%%%%%%%%%%%%%%%%%
\caption{A bundle automaton accepting the full $\{a_1, a_2, a_3, a_4\}$-shift for $\vert \Sigma \vert =2$.}\label{FIGfullshift}
\end{figure}
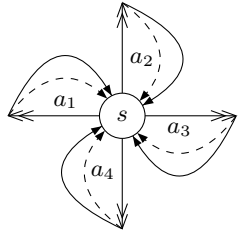

\begin{example}[Monochromatic children]\label{e:AUTmonochromaticchildren}
Consider the unrestricted Rabin automaton $\A = (A,\Sigma,A,\T)$ where the bundle set is given by $\T = \{(a;a;(a_\sigma)_{\sigma \in \Sigma}) \in A \times A \times A^\Sigma : a_\sigma = a_{\sigma'}
\textup{ for all } \sigma, \sigma' \in \Sigma\}$. We then have that $\Sh_\A$ is the
tree shift described in Example~\ref{e:monochromaticchildren}.
If $\vert \Sigma \vert = 2$ and $A = \{a,b\}$ the corresponding automaton is represented in Figure~\ref{FIGmonochromaticchildren}.
%To avoid confusion, we denote by $S_A$ the state set of $\A$, where $S_A$ is a copy of $A$ given by %$S_A = \{s_a : a \in A\}$.
\end{example}
%A homomorphism $\alpha \colon  \Sigma^* \to  A$ accepting the configuration as in Figure~\ref{ALBmonochromaticchildren}, satisfies $\alpha(\varepsilon) = \alpha(00) = \alpha(01) = \alpha(000) = \alpha(001) = \ldots = s_0$ and $\alpha(0) = \alpha(1) = \alpha(10) = \alpha(11) = \alpha(010) = \alpha(011) = \alpha(100) = \alpha(101) = \alpha(110) = \alpha(111) = \ldots = s_1$. This way, the condition~\eqref{eqACCEPTANCE} in Remark~\ref{acceptance} is satisfied. For example,
%$(\alpha(01);f(01);\alpha(010), \alpha(011)) = (s_0; 0; s_1, s_1)$ is a transition bundle of $\A$.

\begin{example}[Even sum]
Let $A = \{0,1\}$.
Consider the unrestricted Rabin automaton $\A = (A,\Sigma,A,\T)$ where $\T = \{(a;a;(a_\sigma)_{\sigma \in \Sigma}) \in A \times A\times A^\Sigma : a+\sum_{\sigma \in \Sigma} a_\sigma \equiv 0 \mod 2\}$.
We have that $\Sh_\A$ is the
tree shift described in Example~\ref{e:zerosummodtwo}.
If $\vert \Sigma \vert = 2$, the corresponding automaton is represented in Figure~\ref{FIGzerosummod2}.
%As in Example~\ref{e:AUTmonochromaticchildren}, we denote by $S_A$ the state set of $\A$.
\end{example}
%A homomorphism $\alpha \colon  \Sigma^* \to  A$ accepting the configuration as in Figure~\ref{ALBzerosummod2}, satisfies $\alpha(\varepsilon) = \alpha(00) = \alpha(11) = \alpha(000) = \alpha(001) =  \alpha(010) = \alpha(101) = \ldots = s_0$ and $\alpha(0) = \alpha(1) = \alpha(01) = \alpha(10) = \alpha(011) = \alpha(100) = \alpha(110) = \alpha(111) = \ldots = s_1$. This way, the condition~\eqref{eqACCEPTANCE} in Remark~\ref{acceptance} is satisfied. For example,
%$(\alpha(01);f(01);\alpha(010), \alpha(011)) = (s_1;1; s_0, s_1)$ is a transition bundle of $\A$.
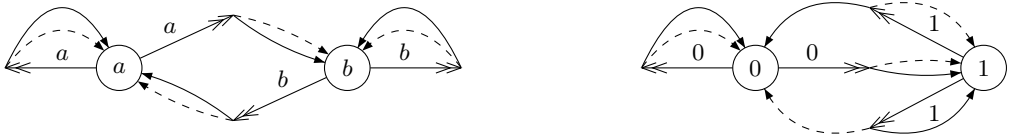
\begin{figure}[h!]
\scriptsize
\centering
\subfigure[An unrestricted Rabin automaton accepting the tree shift of Example~\ref{e:monochromaticchildren} for $\vert \Sigma \vert =2$ and $A = \{a,b\}$.]{
\begin{picture}(70,20)(-35,-13)
\gasset{Nw=6,Nh=6}
\node(A)(-15,0){$a$}
\node(B)(15,0){$b$}
\node[Nframe=n,Nw=-1,Nh=-1](AA)(-30,0){}
\node[Nframe=n,Nw=-1,Nh=-1](AB)(0,7){}
\node[Nframe=n,Nw=-1,Nh=-1](AB_)(0,-7){}
\node[Nframe=n,Nw=-1,Nh=-1](BB)(30,0){}
\drawedge[AHnb=2,AHnb=2,AHangle=20,AHLength=2,AHlength=0,ELside=r](A,AA){$a$}
\drawedge[curvedepth=5,dash={1}0](AA,A){}
\drawedge[curvedepth=8](AA,A){}
\drawedge[AHnb=2,AHnb=2,AHangle=20,AHLength=2,AHlength=0](A,AB){$a$}
\drawedge[curvedepth=1,dash={1}0](AB,B){}
\drawedge[curvedepth=-1](AB,B){}
\drawedge[AHnb=2,AHnb=2,AHangle=20,AHLength=2,AHlength=0](B,BB){$b$}
\drawedge[curvedepth=-5,dash={1}0](BB,B){}
\drawedge[curvedepth=-8](BB,B){}
\drawedge[AHnb=2,AHnb=2,AHangle=20,AHLength=2,AHlength=0,ELside=r](B,AB_){$b$}
\drawedge[curvedepth=1,dash={1}0](AB_,A){}
\drawedge[curvedepth=-1](AB_,A){}
\end{picture}\label{FIGmonochromaticchildren}
}\quad\quad
%%%%%%%%%%%%%%%%%%%%%%%%%%%%%%%%%%%%%%%%%%%%%%%%%%%%%%%%%%%%%%%%%%%%%%%
\subfigure[An unrestricted Rabin automaton accepting the even sum tree shift of Example~\ref{e:zerosummodtwo} for $\vert \Sigma \vert =2$.]{
\begin{picture}(70,20)(-40,-13)
\gasset{Nw=6,Nh=6}
\node(A)(-15,0){$0$}
\node(B)(15,0){$1$}
\node[Nframe=n,Nw=-1,Nh=-1](AA)(-30,0){}
\node[Nframe=n,Nw=-1,Nh=-1](AB)(0,0){}
\node[Nframe=n,Nw=-1,Nh=-1](B^)(0,8){}
\node[Nframe=n,Nw=-1,Nh=-1](B_)(0,-8){}
\drawedge[AHnb=0,ATnb=2,ATangle=20,ATLength=2,ATlength=0](AA,A){0}
\drawedge[curvedepth=5,dash={1}0](AA,A){}
\drawedge[curvedepth=8](AA,A){}
\drawedge[AHnb=2,AHangle=20,AHLength=2,AHlength=0](A,AB){0}
\drawedge[curvedepth=1,dash={1}0](AB,B){}
\drawedge[curvedepth=-1](AB,B){}
\drawedge[AHnb=2,AHangle=20,AHLength=2,AHlength=0,ELside=r](B,B^){1}
\drawedge[curvedepth=4,dash={1}0](B^,B){}
\drawedge[curvedepth=-4](B^,A){}
\drawedge[AHnb=2,AHangle=20,AHLength=2,AHlength=0](B,B_){1}
\drawedge[curvedepth=4,dash={1}0](B_,A){}
\drawedge[curvedepth=-4](B_,B){}
\end{picture}\label{FIGzerosummod2}
}
\caption{Some unrestricted Rabin automata.}
\label{FIGshifts}
\end{figure}

\begin{example}[The even shift on $\Sigma^*$]\label{e:michel}
The unrestricted Rabin automaton $\A$ accepting the even shift
described in Example~\ref{e:evenshift} is represented in Figure~\ref{FIGevenshift}.
For simplicity we limit ourselves to the case $\vert \Sigma \vert = 2$.
%\[
%\begin{array}{cccc}
%(s_0;1;s_0,s_0) &  (s_0;1;s_1,s_0) & (s_0;1;s_0,s_1) & (s_0;1;s_1,s_1)\\
%(s_1;0;s_2,s_2) & & &\\
%(s_2;0;s_1,s_1) & (s_2;0;s_0,s_1) & (s_2;0;s_1,s_0) &
%\end{array}
%\]
%We do not present a figure representing the unrestricted Rabin automaton $\A$ because it would be too complicated to graphically draw it.
%% VEDI SE RIESCI A FARE LA FIGURA. ALTRIMENTI SI POTREBBE AGGIUNGERE IL PROCLAIMER
%%QUI SOTTO
%% Wedo not present a picture/figure representing the unrestricted Rabin automaton $\A$
%% because it would be too complicated to graphically draw it.
\end{example}
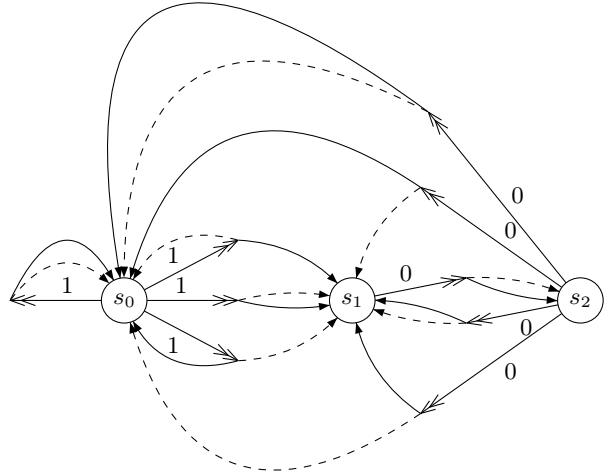
\begin{figure}[h!]
\scriptsize
%%%%%%%%%%%%%%%%%%%%%%%%%%%%%%%%%%%%%%%%%%%%%%%%%%%%%%%%%%%%%%%%%%%%%%%
\begin{center}
\begin{picture}(0,60)(0,-20)
\gasset{Nw=6,Nh=6}
\node(A)(-30,0){$s_0$}
\node(B)(0,0){$s_1$}
\node(C)(30,0){$s_2$}
\node[Nframe=n,Nw=-1,Nh=-1](AA)(-45,0){}
\node[Nframe=n,Nw=-1,Nh=-1](AB)(-15,0){}
\node[Nframe=n,Nw=-1,Nh=-1](BC^)(15,3){}
\node[Nframe=n,Nw=-1,Nh=-1](BC_)(15,-3){}
\node[Nframe=n,Nw=-1,Nh=-1](C^)(9,15){}
\node[Nframe=n,Nw=-1,Nh=-1](C_)(9,-15){}
\node[Nframe=n,Nw=-1,Nh=-1](A^)(-15,8){}
\node[Nframe=n,Nw=-1,Nh=-1](A_)(-15,-8){}
\node[Nframe=n,Nw=-1,Nh=-1](CC)(10,25){}
\drawedge[AHnb=2,AHangle=20,AHLength=2,AHlength=0,ELside=r](A,AA){1}
\drawedge[AHnb=2,AHangle=20,AHLength=2,AHlength=0,ELside=r](C,CC){0}
\drawedge[curvedepth=5,dash={1}0](AA,A){}
\drawedge[curvedepth=8](AA,A){}
\drawedge[AHnb=2,AHangle=20,AHLength=2,AHlength=0](A,AB){1}
\drawedge[curvedepth=1,dash={1}0](AB,B){}
\drawedge[curvedepth=-1](AB,B){}
\drawedge[AHnb=2,AHangle=20,AHLength=2,AHlength=0](B,BC^){0}
\drawedge[curvedepth=1,dash={1}0](BC^,C){}
\drawedge[curvedepth=-1](BC^,C){}
\drawedge[AHnb=2,AHangle=20,AHLength=2,AHlength=0](C,BC_){0}
\drawedge[curvedepth=1,dash={1}0](BC_,B){}
\drawedge[curvedepth=-1](BC_,B){}
\drawedge[AHnb=2,AHangle=20,AHLength=2,AHlength=0](C,C_){0}
\drawedge[curvedepth=15,dash={1}0](C_,A){}
\drawedge[curvedepth=2](C_,B){}
\drawedge[AHnb=2,AHangle=20,AHLength=2,AHlength=0,ELside=r](C,C^){0}
\drawedge[curvedepth=-15](C^,A){}
\drawedge[curvedepth=-2,dash={1}0](C^,B){}
\drawedge[AHnb=2,AHangle=20,AHLength=2,AHlength=0](A,A^){1}
\drawedge[curvedepth=-4,dash={1}0](A^,A){}
\drawedge[curvedepth=2](A^,B){}
\drawedge[AHnb=2,AHangle=20,AHLength=2,AHlength=0,ELside=r](A,A_){1}
\drawedge[curvedepth=4](A_,A){}
\drawedge[curvedepth=-2,dash={1}0](A_,B){}
\drawedge[curvedepth=-20,dash={1}0](CC,A){}
\drawedge[curvedepth=-30](CC,A){}

\end{picture}
\end{center}
%%%%%%%%%%%%%%%%%%%%%%%%%%%%%%%%%%%%%%%%%%%%%%%%%%%%%%%%%%%%%%%%%%%%%%%
\caption{An unrestricted Rabin automaton accepting the even shift for $\vert \Sigma \vert = 2$.}\label{FIGevenshift}
\end{figure}

\subsection{Unrestricted Rabin automata and sofic tree shifts}
\begin{proposition}\label{p:XA sofic} Let $\A = (S,\Sigma,A,\T)$ be an unrestricted Rabin automaton.
Then $\Sh_\A$ is a sofic tree shift.
\end{proposition}
\proof
Consider first the bundle automaton $\mathfrak{A} = (S,\Sigma,\T,\T')$ underlying $\A$. Let $$\F = \{p \in \T^{\Delta_2} : \mathbf{t}_\sigma(p(\varepsilon)) \neq \mathbf{i}(p(\sigma)) \textup{ for some } \sigma \in \Sigma \}.$$
For $f \in \T^{\Sigma^*}$, we have that $f \in \Sh_{\mathfrak{A}}$ if and only if for each $w \in \Sigma^*$ and $\sigma \in \Sigma$ one has $\mathbf{t}_\sigma(f(w)) = \mathbf{i}(f(w\sigma))$ (see Remark~\ref{r:acceptance}).
This is equivalent to $(f^w)\vert_{\Delta_2} \notin \F$ for each $w \in \Sigma^*$, proving that
$$
\Sh_\mathfrak{A} = \Sh(\F).
$$
Hence $\Sh_\mathfrak{A}$ is a tree shift. Moreover, since $\F$ is finite, $\Sh_\mathfrak{A}$ is a tree shift of finite type.

Observe now that the labeling map $\lambda \colon  \T \to  A$ defines a cellular automaton $\tau_\lambda \colon  \T^{\Sigma^*} \to  A^{\Sigma^*}$ given by $\left(\tau_\lambda(f)\right)(w) = \lambda(f(w))$ (the memory set is reduced to $\Delta_1 = \{\varepsilon\}$ and we identify $\T$ with $\T^{\{\varepsilon\}}$). Notice that $\tau_\lambda(\Sh_\mathfrak{A}) = \Sh_\A$. Hence $\Sh_\A$ is a tree shift (see Remark~\ref{r:image}). Moreover, since $\Sh_\mathfrak{A}$ is of finite type, we have that $\Sh_\A$ is sofic.
\endproof

Let $\Sigma$ be a nonempty finite set. Let $A$ and $B$ be two finite alphabets.

\begin{definition}\label{d:def w cdot p}
If $M \subset \Sigma^*$ is a nonempty set and
$p \in A^M$ is a pattern with support $M$, we set $wM = \{wm : m \in M\} \subset \Sigma^*$, where $w \in \Sigma^*$.
We denote by $w \cdot p \in A^{wM}$ the pattern with support $wM$ defined by $(w\cdot p)(wm) = p(m)$ for all $m \in M$.
\end{definition}

\begin{definition}[Unrestricted Rabin automaton associated with a cellular automaton]\label{d:A(tau,M,X)}
Let $X \subset A^{\Sigma^*}$ be a tree shift of finite type and let $\tau \colon  X \to  B^{\Sigma^*}$ be a CA. Let $n \geq 2$ be large enough so that $X$ has memory $n-1$ (see Definition
\ref{d:memory}) and $M = \Delta_n \subset \Sigma^*$ is a memory set for $\tau$
(recall Remark~\ref{r:triangle} as well as the remark following Definition~\ref{d:memory}).
Denote by $\mu \colon  A^M \to  B$ the corresponding local defining map. Set $M' = \Delta_{n-1}$.
The \emph{unrestricted Rabin automaton $\A(\tau,M,X)$ associated with $\tau$} is defined by
$$\A(\tau,M,X) = (X_{M'},\Sigma,B,\T),$$ where $\T \subset X_{M'} \times B \times (X_{M'})^\Sigma$ consists of the bundles $(p;b;(p_\sigma)_{\sigma \in \Sigma})$ such that
\begin{enumerate}[(1)]
\item $p\vert_{\sigma {M'} \cap {M'}}$ equals $(\sigma \cdot p_\sigma)\vert_{\sigma {M'} \cap {M'}}$ for all
$\sigma \in \Sigma$ (that is, $p(\sigma m) = p_\sigma(m)$ whenever $\sigma m \in \sigma {M'} \cap {M'})$;
\item the block $\overline p \colon  M \to  A$, coinciding with $p$ on $M'$ and with $\sigma \cdot p_\sigma$ on $\sigma {M'}$ for all $\sigma \in \Sigma$,  belongs to $X_M$ (such a block $\overline p \in X_M$ is denoted by $\overline{(p;(p_\sigma)_{\sigma \in \Sigma})}$);
\item $b = \mu \left(\overline{(p;(p_\sigma)_{\sigma \in \Sigma})}\right)$.
\end{enumerate}
A transition bundle of $\A(\tau,M,X)$ is illustrated in Figure~\ref{FIGdebruijn} for $\vert \Sigma\vert = 2$.
\end{definition}

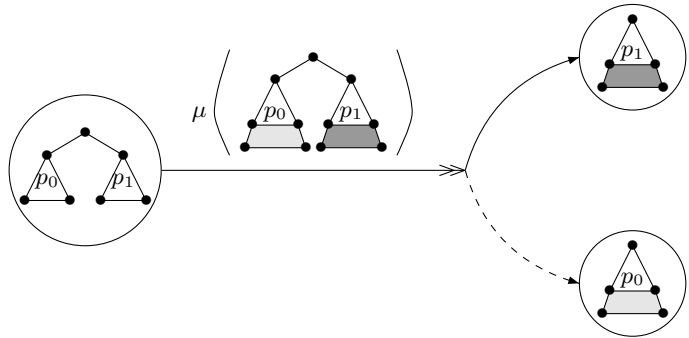
\begin{figure}[h!]
\scriptsize
\centering
%%%%%%%%%%%%%%%%%%%%%%%%%%%%%%%%%%%%%%%%%%%%%%%%%%%%%%%%%%%%%%%%%%%%%%%
\begin{picture}(0,40)(0,-20)
\gasset{Nw=0,Nh=0,Nframe=n,ATnb=0,AHnb=0}
\node(dep)(-20,0){}
\node(arr)(20,0){}
\node(S0)(35,15){}
\node(S1)(35,-15){}
\drawedge[AHnb=2,AHnb=2,AHangle=20,AHLength=2,AHlength=0](dep,arr){}
\drawedge[AHnb=1,curvedepth=3](arr,S0){}
\drawedge[AHnb=1,curvedepth=-3,dash={1}0](arr,S1){}

\node(MU)(-15,7.5){$\mu$}
\drawpolygon[fillcolor=Gray,Nframe=y](-8,6)(-9,3)(-1,3)(-2,6)
\drawpolygon[fillcolor=Gray2,Nframe=y](2,6)(1,3)(9,3)(8,6)
\drawcurve(-11,16)(-13,7.5)(-11,2)
\drawcurve(11,16)(13,7.5)(11,2)
\node(1)(0,15){$\bullet$}
\node(2)(-5,12){$\bullet$}
\node(3)(5,12){$\bullet$}
\node(4)(-5,7.5){$p_0$}
\node(5)(5,7.5){$p_1$}
\node(6)(-8,6){$\bullet$}
\node(7)(-2,6){$\bullet$}
\node(8)(2,6){$\bullet$}
\node(9)(8,6){$\bullet$}
\node(10)(-9,3){$\bullet$}
\node(11)(-1,3){$\bullet$}
\node(12)(1,3){$\bullet$}
\node(13)(9,3){$\bullet$}
\drawedge(1,2){}
\drawedge(1,3){}
\drawedge(2,6){}
\drawedge(2,7){}
\drawedge(3,8){}
\drawedge(3,9){}

\node(A)(-30,5){$\bullet$}
\node(A0)(-35,2){$\bullet$}
\node(A1)(-25,2){$\bullet$}
\node(P0)(-35,-1.5){$p_0$}
\node(P1)(-25,-1.5){$p_1$}
\node(A00)(-38,-4){$\bullet$}
\node(A01)(-32,-4){$\bullet$}
\node(A10)(-28,-4){$\bullet$}
\node(A11)(-22,-4){$\bullet$}
\drawedge(A,A0){}
\drawedge(A,A1){}
\drawedge(A0,A00){}
\drawedge(A0,A01){}
\drawedge(A1,A10){}
\drawedge(A1,A11){}
\drawedge(A00,A01){}
\drawedge(A10,A11){}
\drawcircle[Nframe=y](-30,0,20)

\drawpolygon[fillcolor=Gray2,Nframe=y](39,14)(38,11)(46,11)(45,14)
\node(AA0)(42,20){$\bullet$}
\node(PP0)(42,15.5){$p_1$}
\node(AA00)(39,14){$\bullet$}
\node(AA01)(45,14){$\bullet$}
\node(AA000)(38,11){$\bullet$}
\node(AA010)(46,11){$\bullet$}
\drawedge(AA0,AA00){}
\drawedge(AA0,AA01){}
\drawcircle[Nframe=y](42,15,14)

\drawpolygon[fillcolor=Gray,Nframe=y](39,-16)(38,-19)(46,-19)(45,-16)
\node(AA1)(42,-10){$\bullet$}
\node(PP1)(42,-14.5){$p_0$}
\node(AA10)(39,-16){$\bullet$}
\node(AA11)(45,-16){$\bullet$}
\node(AA100)(38,-19){$\bullet$}
\node(AA110)(46,-19){$\bullet$}
\drawedge(AA1,AA10){}
\drawedge(AA1,AA11){}
\drawcircle[Nframe=y](42,-15,14)

\end{picture}
%%%%%%%%%%%%%%%%%%%%%%%%%%%%%%%%%%%%%%%%%%%%%%%%%%%%%%%%%%%%%%%%%%%%%%%
\caption{A transition bundle of $\A(\tau,M,X)$ when $\vert \Sigma\vert = 2$.}
\label{FIGdebruijn}
\end{figure}

%\begin{remark}\label{r:same M}
%The conditions on $\Delta_n$ in the following proposition are not restrictive.
%Indeed $n$ can be enlarged, if necessary, to still have a memory set for $\tau$ (see %Remark~\ref{r:triangle}) such that $n-1$ is a memory of $X$ (see the remark following %Definition~\ref{d:memory}).
%\end{remark}

%Note that since $M'$ is finite, so is $X_{M'}$ and therefore $\A(\tau,M,X)$ is also
%finite. In other words, $\A(\tau,M,X)$ is an unrestricted Rabin automaton.

\begin{proposition}
\label{p:A(tau,M,X) accepts tau(X)}
Let $X \subset A^{\Sigma^*}$ be a tree shift of finite type with memory $n-1$.
Let $\tau \colon  X \to  B^{\Sigma^*}$ be a cellular
automaton with memory set $\Delta_n$. Then $\Sh_{\A(\tau,\Delta_n,X)} = \tau(X)$.
\end{proposition}
\proof
Set $M = \Delta_n$ and $M' = \Delta_{n-1}$. Suppose that $g \in \Sh_{\A(\tau,M,X)}$. This means that there exists a homomorphism $\alpha \colon  \G_g \to  \A(\tau,M,X)$, that is a map
$\alpha \colon   \Sigma^* \to  X_{M'}$ such that for each $w \in \Sigma^*$, $\sigma \in \Sigma$ and $m \in M'$:
\begin{enumerate}[(1)]
\item $\alpha(w)(\sigma m) = \alpha(w\sigma)(m)$ if $\sigma m \in M'$;\label{compatibility}
\item $\overline{(\alpha(w);(\alpha(w\sigma))_{\sigma \in \Sigma})} \in X_M$
\item $\mu \left(\overline{(\alpha(w);(\alpha(w\sigma))_{\sigma \in \Sigma})}\right) = g(w)$.\label{LABcompatibility}
\end{enumerate}
We define a configuration $f \colon  \Sigma^* \to  A$ by setting
$$f(w) = \alpha(w)(\varepsilon)$$
for each $w \in \Sigma^*$.
Notice that for all $m \in M'$, one has $f(wm) = \alpha(w)(m)$. Indeed if $m = \sigma_1 \cdots \sigma_h \in \Sigma^h$ ($0 \leq h \leq n-2$), we have that $\sigma_i \cdots \sigma_h \in M'$ for each $i = 1, \dots, h$ (see Remark~\ref{r:factor closed}).
By condition~\eqref{compatibility} we have that
\begin{eqnarray*}
f(w m) & = & \alpha(w m)(\varepsilon) = \alpha(w \sigma_1 \sigma_2 \cdots \sigma_{h-1} \sigma_h)(\varepsilon) = \alpha(w \sigma_1 \sigma_2 \cdots \sigma_{h-1})(\sigma_h) = \ldots =\\
%& = &\alpha(w \sigma_1 \ldots \sigma_{h-2})(\sigma_{h-2}\sigma_h) = \\
& = &\alpha(w \sigma_1)(\sigma_2 \cdots \sigma_{h-1} \sigma_h)  = \alpha(w)(\sigma_1 \sigma_2 \cdots \sigma_{h-1} \sigma_h) =\\
& = &\alpha(w)(m).
\end{eqnarray*}
As a consequence, $(f^w\vert_{M'})(m) = f^w(m) = f(wm) = \alpha(w)(m)$ for each $m \in M'$, that is, $f^w\vert_{M'} = \alpha(w) \in X_{M'}$. Hence $f \in X$ (see Remark~\ref{r:B(X) determines X}). Analogously, $f^w\vert_{\sigma M'} = \sigma \cdot \left(\alpha(w\sigma)\right)$ so that $f^w\vert_M = \overline{(\alpha(w);(\alpha(w\sigma))_{\sigma \in \Sigma})}$. By definition, for each $w \in \Sigma^*$, we have
$\tau(f)(w) = \mu(f^w\vert_M)$, thus $\tau(f)(w) = \mu \left(\overline{(\alpha(w);(\alpha(w\sigma))_{\sigma \in \Sigma})}\right) = g(w)$ by condition~\eqref{LABcompatibility}. This proves that $\tau(f) = g$ and then $g \in \tau(X)$.
It follows that $\Sh_{\A(\tau,\Delta_n,X)} \subset \tau(X)$.

For the converse, suppose that $f \in X$. We want to prove that $\tau(f) \in \Sh_{\A(\tau,M,X)}$. The map $\alpha \colon   \Sigma^* \to  X_{M'}$
defined by $\alpha(w) = f^w\vert_{M'}$ yields a homomorphism $\alpha \colon  \G_{\tau(f)} \to  \A(\tau,M,X)$. Indeed the above conditions~\eqref{compatibility}--\eqref{LABcompatibility} are satisfied for each $w \in \Sigma^*$, $\sigma \in \Sigma$ and $m \in M'$:
\begin{enumerate}[(1$'$)]
\item if $\sigma m \in M'$, we have that $\alpha(w)(\sigma m)$ = $(f^w\vert_{M'})(\sigma m)$ = $f^w(\sigma m)$ = $f^{w\sigma}(m)$ = $(f^{w\sigma}\vert_{M'})(m)$ = $\alpha(w\sigma)(m)$;
\item $\overline{(\alpha(w);(\alpha(w\sigma))_{\sigma \in \Sigma})} = f^w\vert_{M} \in X_M$;
\item $\mu \left(\overline{(\alpha(w);(\alpha(w\sigma))_{\sigma \in \Sigma})}\right) = \mu(f^w\vert_M) = \tau(f)(w)$.
\end{enumerate}
This shows that $\tau(f)$ is accepted by $\A(\tau,M,X)$. It follows that $\tau(X) \subset \Sh_{\A(\tau,\Delta_n,X)}$, completing the proof.
\endproof

\begin{remark}\label{r:transducer}
In Proposition~\ref{p:A(tau,M,X) accepts tau(X)} we prove that $\A(\tau,\Delta_n,X)$ is a presentation of $\tau(X)$. In fact, we actually show how to construct a pre-image of any configuration in $\Sh_{\A(\tau,\Delta_n,X)}$. This leads in particular to a presentation of $X$ as well.
%Indeed, it suffices to consider all the transition bundles of the form $(p;p(\varepsilon);(p_\sigma)_{\sigma \in \Sigma})$, where $(p;b;(p_\sigma)_{\sigma \in \Sigma})$ is a bundle of $\A(\tau,\Delta_n,X)$. That is, just by modifying the labeling map of $\A(\tau,\Delta_n,X)$, we obtain an unrestricted Rabin automaton accepting the shift $X$.\\
Consider the bundle automaton $\mathfrak{A} = (X_{M'}, \Sigma, X_M, \T_\mathfrak{A})$, where $\T_\mathfrak{A}$ consists of the transition bundles $(p;\overline{(p;(p_\sigma)_{\sigma \in \Sigma})};(p_\sigma)_{\sigma \in \Sigma})$ such that $p$ and $(p_\sigma)_{\sigma \in \Sigma}$ satisfy conditions (1) and (2) in Definition~\ref{d:A(tau,M,X)}. The tree shift $\Sh_\mathfrak{A}$ is called the \emph{$n$-th higher block shift of $X$}. The presentations of $X$ and $\tau(X)$ are obtained simply by modifying the labels of every transition bundle in $\T_\mathfrak{A}$. More precisely, the transition bundle $\left(p;\overline{(p;(p_\sigma)_{\sigma \in \Sigma})};(p_\sigma)_{\sigma \in \Sigma}\right)$ is replaced by $\left(p;p(\varepsilon);(p_\sigma)_{\sigma \in \Sigma}\right)$ and by $\left(p;\mu\left(\overline{(p;(p_\sigma)_{\sigma \in \Sigma})}\right);(p_\sigma)_{\sigma \in \Sigma}\right)$, respectively.
%This is, in fact, the effect of an application of applying the CA with memory set $M = \{\varepsilon\}$ whose corresponding local defining map
%$\mu'$ is defined by $\mu'(t) = \mu(\overline{(p;(p_\sigma)_{\sigma \in \Sigma}}$ for all transition bundle $t = \left(p;\overline{(p;(p_\sigma)_{\sigma \in \Sigma})};(p_\sigma)_{\sigma \in \Sigma}\right) \in \T_\mathfrak{A}$.

Notice that, being $X_{n-1}$ its state set, each Rabin automaton described above is essential.

This latter construction is a generalization of a de Bruijn graph (these graphs
were introduced by de Bruijn~\cite{deBruijn46} and, independently, by  Good~\cite{Good46}).
\end{remark}

We are now in position to prove the following result. The bottom-up version of it has been proved by Aubrun and B\'eal in~\cite{AubrunBeal12}.

\begin{corollary}\label{c:sofic iff accepted}
A tree shift is sofic if and only
if it is accepted by some unrestricted Rabin automaton.
\end{corollary}
\proof
Let $X \subset B^{\Sigma^*}$ be a sofic tree shift.
%We have to show that $X$ is sofic if and only if it is accepted by some unrestricted Rabin automaton.
By Definition~\ref{d:sofic}, there exist a finite alphabet~$A$, a tree shift of finite type $Y\subset A^{\Sigma^*}$, and a
cellular automaton $\tau \colon  Y \to  B^{\Sigma^*}$ such that $X = \tau(Y)$.
Hence, for a suitable $n\geq 2$ (depending on the memory of $Y$ and on the memory set for
$\tau$) we have $X = \Sh_{\A(\tau, \Delta_n, Y)}$ by Proposition~\ref{p:A(tau,M,X) accepts tau(X)}.
Conversely, if $X = \Sh_\A$ then $X$ is sofic by Proposition~\ref{p:XA sofic}.
\endproof
\subsection{The one-dimensional case}
Let $A$ and $B$ be two finite alphabets. Let $X \subset A^\mathbb{N}$ be a tree shift of finite type and let $\tau \colon  X \to  B^\mathbb{N}$ be a CA.
We can find $n \in \mathbb{N}$ such that the interval $\{0,1,\dots,n-1\}$ is a memory set for $\tau$ and such that $n-1$ is the memory of $X$. Thus, the forbidden blocks of $X$ are words of length $n-1$ and the local
defining map of $\tau$ is defined on $A^n$.

The unrestricted Rabin automaton $\A(\tau, \Delta_n, X)$ has state set $X_{n-1} \subset A^{n-1}$, that is the set of words of length $n-1$ that appear in some configuration of $X$.
Two words $u = a_1a_2 \cdots a_{n-1}$ and $u' = a_1'a_2' \cdots a_{n-1}'$ constitute respectively the source and (unique) terminal states of a transition bundle $t\in \T$ if and only if $a_i' = a_{i+1}$ for all $i=1,2,\dots,n-2$ and $w=a_1a_2 \cdots a_{n-1}a_{n-1}'\in X_n$. Moreover, if $\mu \colon  A^n \to  B$ is the corresponding local defining map for $\tau$, then $\lambda(t) = \mu(w)$.
Hence the configurations in $\tau(X)$ are the right-infinite sequences of labels that correspond to a right-infinite path on this directed graph.

\begin{example}\label{e:even1dimensional}
Let $Y = \Sh(\{11\}) \subset \{0,1\}^\mathbb{N}$ be the golden mean shift and let $X = \Sh(\{10^{2n+1}1 : n \in \mathbb{N}\})\subset \{0,1\}^\mathbb{N}$ be the even shift. These two shifts are both presented in Example~\ref{e:evenshift}. In that example we also defined a surjective cellular automaton $\tau \colon  Y \to X$ with memory set $M  = \{0,1\}= \Delta_2$.
The bundle automaton $\mathfrak{A}$ accepting the $2$nd higher block shift of $Y$ is represented in Figure~\ref{FIGgoldeneven}. As pointed out in Remark~\ref{r:transducer}, by labeling each edge of $\mathfrak{A}$ with the first letter of the word corresponding to the source state of the edge, we get a presentation of the golden mean shift (see Figure~\ref{FIGgolden}). By labeling each edge of $\mathfrak{A}$ with its image under $\mu$, we get the unrestricted Rabin automaton $\A = \A(\tau, M, Y)$ associated with~$\tau$. Proposition~\ref{p:A(tau,M,X) accepts tau(X)} ensures that $\A$ is a presentation of the even shift (see Figure~\ref{FIGeven}).
\end{example}

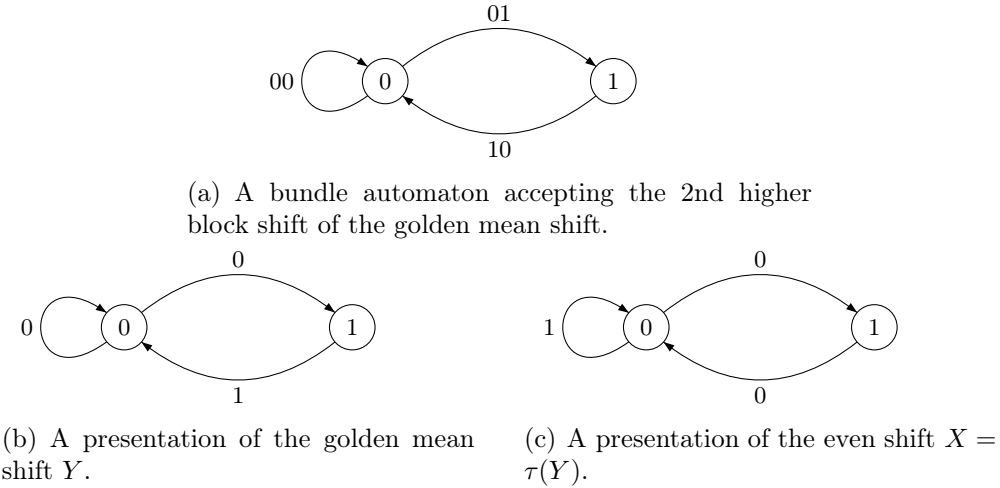
\begin{figure}[h!]
\scriptsize
\centering
\subfigure[A bundle automaton accepting the 2nd higher block shift of the golden mean shift.]{
%%%%%%%%%%%%%%%%%%%%%%%%%%%%%%%%%%%%%%%%%%%%%%%%%%%%%%%%%%%%%%%%%%%%%%%
\begin{picture}(80,20)(-40,-12)
\gasset{Nw=6,Nh=6}
\node(A)(-15,0){0}
\node(B)(15,0){1}
\drawloop[loopangle=180](A){00}
\drawedge[curvedepth=7](A,B){01}
\drawedge[curvedepth=7](B,A){10}
\end{picture}\label{FIGgoldeneven}
}\\
%%%%%%%%%%%%%%%%%%%%%%%%%%%%%%%%%%%%%%%%%%%%%%%%%%%%%%%%%%%%%%%%%%%%%%%
\subfigure[A presentation of the golden mean shift~$Y$.]{
\begin{picture}(60,20)(-30,-12)
\gasset{Nw=6,Nh=6}
\node(A)(-15,0){0}
\node(B)(15,0){1}
\drawloop[loopangle=180](A){0}
\drawedge[curvedepth=7](A,B){0}
\drawedge[curvedepth=7](B,A){1}
\end{picture}\label{FIGgolden}
}\quad\quad
%%%%%%%%%%%%%%%%%%%%%%%%%%%%%%%%%%%%%%%%%%%%%%%%%%%%%%%%%%%%%%%%%%%%%%%
\subfigure[A presentation of the even shift~$X=\tau(Y)$.]{
\begin{picture}(60,20)(-30,-12)
\gasset{Nw=6,Nh=6}
\node(A)(-15,0){0}
\node(B)(15,0){1}
\drawloop[loopangle=180](A){1}
\drawedge[curvedepth=7](A,B){0}
\drawedge[curvedepth=7](B,A){0}
\end{picture}\label{FIGeven}
}
\caption{Representation of a surjective cellular automaton $\tau$ from the golden mean shift $Y$ onto the even shift $X$.}
\label{FIGevene}
\end{figure}

\subsection{Regular configurations and surjunctivity of sofic tree shifts}\label{ss:Regular configurations in sofic tree shifts}
Let $\Sigma$ be a nonempty finite set. Let $A$ and $B$ be two finite alphabets.

\begin{definition}\label{d:regular}
A configuration $f\in A^{\Sigma^*}$ is said to be \emph{regular} (or \emph{periodic}) if its orbit under the action of $\Sigma^*$ is finite, that is, $\{f^w : w \in \Sigma^*\}$ is a finite set. %If $\vert \{f^w : w \in \Sigma^*\}\vert = n$, we say that $f$ is \emph{$n$-regular}.
\end{definition}

\begin{theorem}
\label{t:regular}
Let $X \subset A^{\Sigma^*}$ be a sofic tree shift. Then the set of regular configurations in $X$ is dense in $X$ with respect to the prodiscrete topology.
\end{theorem}
\proof
Let $\A = (S,\Sigma,A,\T)$ be an essential presentation of $X$ (see Remark~\ref{r:essential}).
Fix a map $\xi \colon  S \to  \T$ that associates with each state $s\in S$ a transition bundle $\xi(s) \in \T$ starting at $s$.
For each $s\in S$ there exists exactly one homomorphism $\alpha_s \colon  \Sigma^* \to S$ such that $\alpha_s(\varepsilon) = s$ and
$\alpha_s(w\sigma) = \mathbf{t}_\sigma(\xi(\alpha_s(w)))$, for every $w \in \Sigma^*$ and $\sigma \in \Sigma$.
Let $f_s \in X$ be the configuration accepted by $\A$ via $\alpha_s$. Notice that for each $s \in S$ and $w \in \Sigma^*$ one has $(\alpha_s)^w = \alpha_{\alpha_s(w)}$. Thus $(f_s)^w = f_{\alpha_s(w)}$. Moreover, since $\vert \{f_s : s \in S\} \vert \leq \vert S \vert$, each configuration $f_s$ is regular.

Fix $f \in X$, and suppose that $f$ is accepted by $\A$ via a homomorphism $\beta \colon \Sigma^* \to S$.
Define $\beta_n \colon  \Sigma^* \to  S$ as the map coinciding with $\beta$ on $\Delta_n$ and such that $(\beta_n)^w = \alpha_{\beta(w)}$ for each $w \in \Sigma^{n-1}$.
It is easy to see that there exists a suitable configuration $g_n\in A^{\Sigma^*}$ such that $\beta_n \colon  \G_{g_n} \to  \A$ is a homomorphism.
By definition, we have that $g_n\in X$.

Note that $g_n$ is regular. Indeed for each $w \in \Sigma^{n-1}$ and $v \in \Sigma^*$ we have $(g_n)^{wv} = ((g_n)^w)^v = (f_{\beta(w)})^v = f_{\alpha_{\beta(w)}(v)}$.
Since $\lim_{n \to \infty}g_n =f$, we have that $f$ is in the closure of the regular configurations of $X$.
\endproof
\begin{remark}
The above result is no more true, in general, when the free monoid $\Sigma^*$ is replaced by a free group.
In~\cite{Fiorenzi09}, a counterexample is illustrated in which
$A = \{0,1\}$ and $X \subset A^\mathbb{Z}$ is the subshift of finite type with
defining set of forbidden words ${\mathcal F} = \{01\}$.
\end{remark}

Given any subset $X \subset A^{\Sigma^*}$ we denote by $\mathcal{R}(X)$
(resp. $\mathcal{R}_n(X)$) the set of regular configurations of $X$
(resp. regular configurations of $X$ whose orbits have at most
$n$ elements).

In order to state the following characterization of regularity of
configurations in $A^{\Sigma^*}$ we recall that an equivalence relation
$R \subset S \times S$ on a semigroup $S$ is said to be \emph{right-invariant}
 provided that $(s_1,s_2) \in R$ implies $(s_1s, s_2s) \in
R$ for all $s_1,s_2,s \in S$. We also say that $R$ is of \emph{finite index} if there are only finitely many $R$-equivalence
classes $[s] := \{s' \in S: (s,s') \in R\}$, where $s \in S$. In this case, the number of these classes is called the \emph{index} of $R$.

\begin{proposition}\label{p:regular iff constant}
Let $f \in A^{\Sigma^*}$ be a configuration. Then the following conditions are equivalent:
\begin{enumerate}[{\rm (a)}]
\item $f$ is regular;
\item there exists a right-invariant equivalence relation $R$ of finite index on $\Sigma^*$
 such that $f$ is constant on each equivalence class of $R$.
\end{enumerate}
\end{proposition}

\begin{proof}
Let $f \in A^{\Sigma^*}$ and consider the map
$\rho_f \colon \Sigma^* \to  A^{\Sigma^*}$ defined by
$\rho_f(w) = f^w$ for all $w \in \Sigma^*$.
Then $\rho_f$ induces the equivalence relation $R_f$ on
$\Sigma^*$ defined by $(w,w') \in R_f$ if $\rho_f(w) = \rho_f(w')$
(i.e., if $f^w = f^{w'}$), where $w,w' \in \Sigma^*$. Moreover $\rho_f$ induces a bijection between the
set of $R_f$-equivalence classes in $\Sigma^*$ and the
image of $\rho_f$ (the latter coincides with the orbit of the configuration $f$).

Suppose that $w_1,w_2 \in \Sigma^*$ and $(w_1, w_2) \in R_f$, i.e., $f^{w_1} = f^{w_2}$.
Then $f(w_1) = f^{w_1}(\varepsilon) = f^{w_2}(\varepsilon) = f(w_2)$, thus showing that
$f$ is constant on each $R_f$-equivalence class.
Moreover, for $w, u \in \Sigma^*$, we have $f^{w_1w}(u) = f^{w_1}(wu) = f^{w_2}(wu) = f^{w_2w}(u)$, so that
$(w_1w,w_2w) \in R_f$ as well. This shows that the equivalence relation $R_f$ is right-invariant.

We deduce that
\begin{equation}
\label{eq:regular-iff-constant}
f \textup{ is regular if and only if } R_f\textup{ is of finite index.}
\end{equation}
The implication ``(a) $\Longrightarrow$ (b)'' is now straightforward.

Conversely, suppose (b) and suppose that $w_1,w_2,w \in \Sigma^*$ and $(w_1, w_2) \in R$. We then have $f^{w_1}(w) = f(w_1w) = f(w_2w) = f^{w_2}(w)$ (where the middle equality follows from the right-invariance of $R$ and the fact that $f$ is constant on each $R$-equivalence class), thus showing that $(w_1, w_2) \in R_f$.
It follows that $R \subset R_f$ so that $R_f$ also is of finite index.
From \eqref{eq:regular-iff-constant} we deduce that $f$ is regular.
\end{proof}

\begin{remark}\label{r:regular iff constant}
The proof of Proposition~\ref{p:regular iff constant} actually shows that a configuration $f \in A^{\Sigma^*}$ belongs to $\mathcal{R}_n(A^{\Sigma^*})$ if and only if
there exists a right-invariant equivalence relation $R$ of finite index $\leq n$ on $\Sigma^*$
such that $f$ is constant on each equivalence class of $R$.
\end{remark}

\begin{corollary}\label{c:Rn finite}
The set $\mathcal{R}_n(A^{\Sigma^*})$
is finite for every $n \in \mathbb{N}$.
\end{corollary}

\begin{proof}
Let $R \subset \Sigma^* \times \Sigma^*$ be a right-invariant equivalence relation of index $n$.
Given $w \in \Sigma^*$ we denote by $\mathfrak{t}(w)$ the $R$-equivalence class of $w$ and we call it the \emph{$R$-type} of $w$.
Thus, there are exactly $n$ many $R$-types.
Pick $w_1, w_2, \dots, w_n \in \Sigma^*$ representatives of these $R$-classes of minimal word length.  We claim that $|w_i| < n$ for all $i=1,2,\dots, n$. Indeed suppose by
contradiction that one of these representatives, denote it by $w$, satisfies
$w = \sigma_1\sigma_2 \cdots \sigma_k$ where $\sigma_i \in \Sigma$, $1 \leq i \leq k$ and $k \geq n$.
Set $u_0 = \varepsilon$ and $u_i = \sigma_1\sigma_2 \cdots \sigma_i$ for $1
\leq i \leq k$. Then necessarily there exist $0 \leq i < j \leq k$ such that
$(u_i, u_j) \in R$. But then $w' = \sigma_1\sigma_2 \cdots \sigma_i
\sigma_{j+1} \cdots \sigma_k$ would satisfy $(w', w) \in R$ and $|w'|<|w|$, contradicting
the minimality of $|w|$.

It follows from our discussion that the
equivalence relation $R$ is recursively (from top to down) recovered by the following finite data:
\begin{equation}
\label{eq:finite-data}
\left(\mathfrak{t}(w_i), (\mathfrak{t}(w_i\sigma))_{\sigma \in \Sigma}\right)_{i=1}^n.
\end{equation}

Since the data \eqref{eq:finite-data} can be detected by looking at the $R$-types of
the elements just in $\Delta_{n+1}$, we deduce that there are at most $n^{|\Delta_{n+1}|}$
many such right-invariant equivalence relations of index $n$ on $\Sigma^*$.

Let $f \in A^{\Sigma^*}$ be a configuration which is constant
on each equivalence class of a suitable right-invariant equivalence relation $R$ of finite index $n$ on $\Sigma^*$.
Then $f$ only depends on its values on the $n$ representatives of $R$.
It follows that there are at most $n^{|\Delta_{n+1}|}|A|^n$ many such regular configurations $f$
in $A^{\Sigma^*}$.

By Remark~\ref{r:regular iff constant}, it follows that $\vert \mathcal{R}_n(A^{\Sigma^*}) \vert \leq \sum_{i=1}^n i^{|\Delta_{i+1}|}|A|^i \leq n^{|\Delta_{n+1}|+1}|A|^n$.
\end{proof}

\begin{proposition}\label{p:dense under tau}
Let $X \subset A^{\Sigma^*}$ be a tree shift and let $\tau \colon X \to
B^{\Sigma^*}$ be a CA.
If $\mathcal{R}(X)$ is dense in $X$ then $\mathcal{R}(\tau(X))$
is dense in $\tau(X)$.
\end{proposition}

\begin{proof}
Set $Y = \tau(X)$. First we prove that $\tau(\mathcal{R}(X))
\subset \mathcal{R}(Y)$. This follows from the fact that $\tau$ commutes
with the shift action (Theorem~\ref{t:curtis}),
indeed if $f \in X$ and $(w,w')\in R_f$ then $f^w = f^{w'}$ and $(\tau(f))^w = \tau(f^w) = \tau(f^{w'}) = (\tau(f))^{w'}$ that is
$(w,w') \in R_{\tau(f)}$. Moreover, if $R_f$ has
finite index, then $R_{\tau(f)}$ has finite index as well.
Finally, from the continuity of $\tau$, we deduce $Y = \tau(X) = \tau(\overline{\mathcal{R}(X)}) \subset \overline{\tau(\mathcal{R}(X))} \subset \overline{\mathcal{R}(Y)}$.
\end{proof}

\begin{definition}[Surjunctivity]
A selfmapping $\tau : X \to X$ on a set $X$ is
\emph{surjunctive} if it is either noninjective or surjective.
\end{definition}
In
other words a map is surjunctive if it is not a strict
embedding. Hence the implication ``injective $\Longrightarrow$
surjective'' holds for surjunctive maps. The notion of surjunctivity is due to
Gottschalk~\cite{Gottschalk73}.

The simplest example of a surjunctive map is provided by a
selfmapping $f \colon X \to X$ where $X$ is a finite set. Other examples
are provided by linear selfmappings of finite-dimensional vector spaces
and by regular selfmappings of complex algebraic varieties (and, more
generally, of algebraic varieties over algebraically closed fields).
This latter is a highly nontrivial result called the Ax-Grothendieck Theorem (see \cite{Ax68}). Many others examples of
surjunctive maps are given by Gromov in~\cite{Gromov99} (see also~\cite{CeccheriniCoornaert13}).
Moreover, Richardson proves in \cite{Richardson72} that a cellular automaton
$\tau \colon A^{\mathbb{Z}^d} \to A^{\mathbb{Z}^d}$ is surjunctive for each $d \geq 1$.
%More generally, Lawton~\cite{Lawton72}
%independently proved that every cellular automaton $\tau \colon A^G \to
%A^G$ is surjunctive where $G$ is any residually finite group.
In fact surjunctivity of cellular
automata $\tau \colon A^G \to A^G$ was proved for every \emph{amenable}
group $G$ as a consequence of the Garden of Eden Theorem for amenable
groups (\cite{CeccheriniMachiScarabotti99}, see also \cite[Theorem 5.9.1]{livre}) and, more
generally, for every \emph{sofic} group (a result due to Gromov
\cite{Gromov99} and Weiss \cite{Weiss00}, see also \cite[Theorem
7.8.1]{livre}).

The following is a sufficient condition for a selfmapping of a
topological space to be surjunctive. Similar conditions are stated
in~\cite{Gromov99}.

\begin{lemma}\label{density}
Let X be a topological space, let $\tau : X \to X$ be a
closed map and let $(X_i)_{i \in I}$ be a family of subsets of
$X$ such that
\begin{itemize}
\item $X = \overline{\bigcup_{i \in I}{X_i}}$
\item $\tau(X_i) \subseteq X_i$
\item $\tau\vert_{X_i} : X_i \to X_i$ is surjunctive
\end{itemize}
then $\tau$ is surjunctive.
\end{lemma}

\begin{proof}
If $\tau$ is injective then, for every $i \in I$, the restriction
$\tau\vert_{X_i}$ is injective as well. By the hypotheses we have
$\tau(X_i) = X_i$ and hence $\bigcup_{i \in I}X_i = \bigcup_{i \in
I}\tau(X_i) = \tau(\bigcup_{i \in I}X_i) \subseteq
\tau(\overline{\bigcup_{i \in I}X_i}) = \tau(X)$. Then $X =
\overline{\bigcup_{i \in I}X_i} \subseteq \overline{\tau(X)}$, and
$\tau$ being closed we have $X \subseteq \tau(X)$.
\end{proof}

In the following theorem we prove that the density of regular
configurations is a sufficient condition for the surjunctivity of a
CA defined on a tree shift.

\begin{theorem}\label{t:dense -> surjunctive}
Let $X \subseteq A^{\Sigma^*}$ be a tree shift whose set
$\mathcal{R}(X)$ of regular configurations is dense in $X$. Then
every cellular automaton $\tau : X \to X$ is surjunctive.
\end{theorem}

\begin{proof}
By Corollary~\ref{c:Rn finite}, the set $\mathcal{R}_n(X)
= \mathcal{R}_n(A^{\Sigma^*}) \cap X$ is finite. As proved in
Proposition~\ref{p:dense under tau}, we have that
$\tau(\mathcal{R}_n(X)) \subset \mathcal{R}_n(\tau(X)) \subset \mathcal{R}_n(X)$. Hence $\tau$ is
surjunctive by Lemma~\ref{density}.
\end{proof}

From Theorem~\ref{t:regular} we then deduce the following result.

\begin{corollary}\label{c:surjunctive}
Let $X \subset A^{\Sigma^*}$ be a sofic tree shift. If $\tau : X \to X$ is a cellular
automaton, then $\tau$ is surjunctive.
\end{corollary}

\begin{remark}The implication ``injective $\Longrightarrow$
surjective'' in Theorem~\ref{t:dense -> surjunctive} is not invertible. The following is an
example when $\vert \Sigma \vert =1$ and $A = \{ 0,1 \}$. Let $\tau$ be the cellular automaton given by the
local defining map $\mu : A^{\Delta_3} \to A$ such
that
$$\mu(a_0, a_1, a_2) = a_0 + a_2 \mod 2.$$
The cellular automaton $\tau$ is surjective and not injective. Indeed if
$(a_n)_{n \in \mathbb{N}}$ is a configuration in
$A^\mathbb{N}$, a pre-image of $(a_n)_{n \in \mathbb{N}}$ is given by:
$$
\left\{
\begin{array}{ll}
b_0 = 0&\\
b_1 = 0&\\
b_n = a_{n-1}+b_{n-2}\mod2 & \textup{if } n \geq 2\\
\end{array}
\right.
$$
as showed in Figure~\ref{fig:image}.
\begin{figure}[!h]
%%%%%%%%%%%%%%%%%%%%%%%%%%%%%%%%%%%%%%%%%%%%%%%%%%%%%%%%%%%%%%%%%%%%%%%
\scriptsize
\begin{center}
\begin{tabular}{r|c|c|c|c|c|c|c|c|c|c|c|}
\cline{2-11}
$(b_n)_{n \in \mathbb{N}}$ & $0$ & $0$ & $a_0$ & $a_1$ & $a_0+a_2$ & $a_1+a_3$ & $a_0+a_2+a_4$ & $a_1 + a_3 +a_5$ & $a_0+a_2+a_4+a_6$ &\dots \\
\cline{2-11}
$(a_n)_{n \in \mathbb{N}}$ & $a_0$ & $a_1$ & $a_2$ & $a_3$ & $a_4$ & $a_5$ & $a_6$ & $a_7$ & $a_8$ &\dots \\
\cline{2-11}
\end{tabular}\ .
\end{center}
\caption{The image of the configuration $(b_n)_{n \in \mathbb{N}}$ under the cellular automaton $\tau$.}
\label{fig:image}
\end{figure}
By taking $(b_n+1)_{n \in \mathbb{N}}$ we get a different pre-image.
\end{remark}

\section{Full-tree-patterns of a sofic tree shift}\label{s:Full-tree-patterns}
\begin{definition}[Sub-bundle]
Let $\A =(S,\Sigma,A,\T)$ be an unrestricted Rabin automaton. Let $M \subset \Sigma$ be a subset. A tuple $(s;a;(s_\sigma)_{\sigma \in M}) \in S \times A \times S^M$ is called a \emph{sub-bundle} of a transition bundle $(\bar s;\bar a;(\bar s_\sigma)_{\sigma \in \Sigma})\in \T$ provided that $s = \bar s$, $a = \bar a$, and $s_{\sigma} = \bar s_\sigma$ for each $\sigma \in M$.
%The label $\lambda(\bar s;(\bar s_\sigma)_{\sigma \in \Sigma})$ is naturally associated with the sub-bundle $(s;(s_\sigma)_{\sigma \in M})$.
\end{definition}

Recall that a \emph{$k$-ary rooted tree} is a rooted tree in which each vertex has at most $k$ children. A \emph{leaf} is a vertex without children. A \emph{full $k$-ary rooted tree} is a $k$-ary rooted tree in which every vertex other than the leaves has exactly $k$ children. Hence $\Sigma^*$ is the full $\vert\Sigma\vert$-ary rooted tree with no leaves. A \emph{subtree of $\Sigma^*$} is a connected subgraph of $\Sigma^*$ containing the root $\varepsilon$. Thus, a subtree of $\Sigma^*$ is always a $\vert\Sigma\vert$-ary rooted tree.
The \emph{height} of a finite subtree $T \subset \Sigma^*$ as the minimal $n \geq 1$ such that $T \subset \Delta_n$.
In other words, the height of a subtree is the number of vertices contained in a maximal path from the root to a leaf (we call such a path a \emph{branch}).

If $T\subset\Sigma^*$ is a subtree and $w \in T$,
we denote by $\Sigma_T(w)$ the set  $\{\sigma \in \Sigma : w\sigma \in T\}$. Hence $w\in T$ is a leaf if and only if $\Sigma_T(w) = \varnothing$.

Given a subtree $T\subset\Sigma^*$, we denote by $T^+$ the subtree of $\Sigma^*$ defined as $T \cup \{w\sigma : w \in T, \sigma \in \Sigma\}$. Note that $T^+$ is always a full subtree. If $T$ is a full subtree, then $T^+$ is obtained by adding all the $k$ children of each leaf in $T$. In particular, for each $n\geq1$ the set $\Delta_n$ is a full subtree of $\Sigma^*$ whose leaves are the elements in $\Sigma^{n-1}$. Moreover, $\Delta_n^+ = \Delta_{n+1}$.

Finite full subtrees correspond to finite and complete prefix codes in~\cite{Aubrun11}.

\begin{definition}\label{d:acceptanceSUBTREE}
Let $\A =(S,\Sigma,A,\T)$ be an unrestricted Rabin automaton. Let $T \subset \Sigma^*$ be a subtree and let $f \colon T \to A$ be a map. One
says that $f$ is \emph{accepted by $\A$} if there exists a map $\alpha \colon  T \to  S$ such that, for each $w \in T$, one has
\begin{equation}\label{eq:acceptanceSUBTREE}
(\alpha(w);f(w);(\alpha(w\sigma))_{\sigma \in \Sigma_T(w)}) \textup{ is a sub-bundle of some } t \in \T.
\end{equation}
In this case we say that $f$ is accepted by $\A$ \emph{via} $\alpha$.
\end{definition}

Note that, for a leaf $w \in T$, condition \eqref{eq:acceptanceSUBTREE} reduces to saying that there exists a transition bundle starting at $\alpha(w)$ with label $f(w)$ (in fact, $\alpha$ is not defined on $w\sigma$ for any $\sigma\in \Sigma$).

\smallskip

Let $\A$ be an unrestricted Rabin automaton and let $\Sh_\A \subset A^{\Sigma^*}$ be the sofic tree shift accepted by $\A$. If $T \subset \Sigma^*$ is a subtree and $f \in X_T$ then obviously $f$ is accepted by $\A$. If $\A$ is essential, also the converse of this fact holds.

\begin{theorem}\label{t:subtree} Let $\A =(S,\Sigma,A,\T)$ be an essential unrestricted Rabin automaton. Let $T \subset \Sigma^*$ be a subtree and suppose that $f \in A^T$ is accepted by $\A$. Then there exists a configuration $\bar f \in \Sh_\A$ such that $f = \bar f\vert_T$.
\end{theorem}
\proof
Suppose that $f$ is accepted by $\A$ via $\alpha \colon  T \to  S$. We define a homomorphism $\bar \alpha \colon  \Sigma^* \to  S$ and a configuration $\bar f \colon  \Sigma^* \to  A$ in such a way that $\bar \alpha\vert_T = \alpha$, $\bar f\vert_T = f$, and with $\bar \alpha$ and $\bar f$ satisfying condition~\eqref{eq:acceptance} of Remark~\ref{r:acceptance}. For this,
we recursively define $\bar \alpha$ on $\Delta_n$ for each $n \geq 1$.

If $n = 1$ then $\Delta_1 = \{\varepsilon\}$ and we obviously set $\bar \alpha(\varepsilon) = \alpha(\varepsilon)$.

Suppose now that $\bar \alpha$ has been defined on $\Delta_n$. For each $w \in \Sigma^{n-1}$ and $\sigma \in \Sigma$, we have to define $\bar \alpha(w\sigma)$.

Consider first the case $w \in T$. Then there exists a bundle $t \in \T$ such that $(\alpha(w);f(w);$ $(\alpha(w\sigma))_{\sigma \in \Sigma_T(w)})$ is a sub-bundle of $t$ and $\lambda(t)=f(w)$. If $t = (s;a;(s_\sigma)_{\sigma \in \Sigma})$,
we define $\bar \alpha(w\sigma) = s_{\sigma}$. Notice that $\bar \alpha$ and $\alpha$ coincide on $w\Sigma_T(w)$. Moreover, $(\bar \alpha(w);f(w);$ $(\bar \alpha(w\sigma))_{\sigma \in \Sigma}) \in \T$.

If $w\notin T$, consider $s= \bar \alpha(w) \in S$. Since $\A$ is essential, we can fix a transition bundle $(s;a;(s_\sigma)_{\sigma \in \Sigma})\in \T$ starting at $s$. Notice that $a = a(w)$ and $s_\sigma=s_\sigma(w)$. By defining $\bar \alpha(w\sigma) = s_\sigma$, we have that $(\bar \alpha(w);a(w);$ $(\bar \alpha(w\sigma))_{\sigma \in \Sigma})$ $\in \T$.

Define now $\bar f \colon  \Sigma^* \to  A$ by setting $\bar f\vert_T = f$ and $\bar f(w) = a(w)$ for $w\notin T$. By definition (see Remark~\ref{r:acceptance}), the configuration $\bar f$ is accepted by $\A$ via the homomorphism $\bar \alpha$.
\endproof

\begin{corollary}\label{c:subtree}
Let $\A =(S,\Sigma,A,\T)$ be an essential unrestricted Rabin automaton. Let $T \subset \Sigma^*$ be a subtree and let $f \in A^T$ be a map. Then
$$f \textup{ \it is accepted by } \A \Longleftrightarrow f \in X_T.$$
\end{corollary}

\begin{definition}
Let $T$ be a finite full subtree of $\Sigma^*$. A pattern defined on $T$ is called a \emph{full-tree-pattern}.
The set of all full-tree-patterns is denoted by $\TT(A^{\Sigma^*})$. Given a shift $X \subset A^{\Sigma^*}$,
we denote by $\TT(X)$ the set of all full-tree-patterns of $X$ (that is, $\TT(X)= \bigcup_{T \subset \Sigma^*} X_T$, where the union ranges over all finite full subtrees $T$ of $\Sigma^*$). The \emph{height of a full-tree-pattern $p \in A^T$} as the height of its support, i.e., of the (finite full) subtree $T$.
\end{definition}

In this setting, we have the following characterization of acceptance which immediately
results from Definition~\ref{d:acceptanceSUBTREE}.

\begin{proposition}\label{p:acceptanceSUBTREE2}
Let $\A =(S,\Sigma,A,\T)$ be an unrestricted Rabin automaton. Let $T \subset \Sigma^*$ be a finite full subtree. A full-tree-pattern $p \in A^T$ is accepted by $\A$ if and only if there exists a map $\alpha \colon  T^+ \to  S$ such that
\begin{equation}\label{eq:acceptanceSUBTREE2}
(\alpha(w);p(w);(\alpha(w\sigma))_{\sigma \in \Sigma})\in \T
\end{equation}
for each $w \in T$.
\end{proposition}

By abuse of language, if \eqref{eq:acceptanceSUBTREE2} holds and there is no ambiguity, we say that the full-tree-pattern \emph{$p$ is accepted by $\A$ via $\alpha$}.
Obviously, Proposition~\ref{p:acceptanceSUBTREE2} applies whenever $T = \Delta_n$ for some $n \geq 1$ (recall that in this case $T^+ = \Delta_{n+1}$).

The following result follows from Corollary~\ref{c:subtree}.
\begin{corollary}\label{c:full-subtree}
Let $\A =(S,\Sigma,A,\T)$ be an essential unrestricted Rabin automaton. Let $p\in \TT(A^{\Sigma^*})$ be a full-tree-pattern. Then
$$p \textup{ \it is accepted by } \A \Longleftrightarrow p\in \TT(\Sh_\A).$$
\end{corollary}

To conclude this section, we have the following (see Remark~\ref{r:B(X) determines X}).
\begin{remark}\label{r:T(X) determines X}
Let $X,Y \subset A^{\Sigma^*}$ be two tree shifts. Then
$$X = Y \Longleftrightarrow \TT(X) = \TT(Y).$$
\end{remark}

\section{Irreducibility}
\begin{definition}[Bundle path]\label{d:bundle path}
Let $\A =(S,\Sigma,A,\T)$ be an unrestricted Rabin automaton. A \emph{bundle path} in $\A$ is a sequence $\pi = (t_0,t_1,\dots, t_{n-1})$ of transition bundles such that there exist $\sigma_0, \sigma_1,\dots, \sigma_{n-1} \in \Sigma$ satisfying
$$
{\bf t}_{\sigma_i}(t_{i}) = {\bf i}(t_{i+1})
$$for each $i = 0, \dots, n-2$.
We then say that $\pi$ \emph{begins} at ${\bf i}(t_{0}) \in S$ and
\emph{ends} at ${\bf t}_{\sigma_{n-1}}(t_{n-1})\in S$.
% When no ambiguity occurs, we also say that the bundle path \emph{ends} at ${\bf t}(t_{n-1}) \in S^\Sigma$.
%$$(s^{(0)};a_0;(s^{(0)}_\sigma)_{\sigma \in \Sigma}),\ %(s^{(1)};a_1;(s^{(1)}_\sigma)_{\sigma \in \Sigma}),\ \ldots,\ %(s^{(n-1)};a_{n-1};(s^{(n-1)}_\sigma)_{\sigma \in \Sigma})$$
%such that for each $i = 0, \ldots n-2$, there exists $\sigma_i \in \Sigma$ verifying
%$s^{(i+1)} = s^{(i)}_{\sigma_i}$. In this case we say that the bundle path \emph{starts} at $s^{(0)}$ and \emph{ends} at $s^{(n-1)}_\sigma$, for each $\sigma \in \Sigma$. When no ambiguity is possible we also say that the bundle path \emph{ends} at the sequence $(s^{(n-1)}_\sigma)_{\sigma \in \Sigma}$.
The \emph{label} of the bundle path $\pi$ is the word
%$$\lambda(s^{(0)};(s^{(0)}_\sigma)_{\sigma \in \Sigma})\ \lambda(s^{(1)};(s^{(1)}_\sigma)_{\sigma \in \Sigma})\ \ldots\ \lambda(s^{(n-1)};(s^{(n-1)}_\sigma)_{\sigma \in \Sigma})$$
$$\lambda(\pi) = \lambda(t_0)\lambda(t_1) \cdots \lambda(t_{n-1}) \in A^*.$$
The integer $n\geq0$ is the \emph{length} of $\pi$.
If $n=0$ the bundle path reduces to a state $s\in S$. In this case it is
called the \emph{empty path} at $s$, and its label is the empty word $\varepsilon$.
\end{definition}
Note that a bundle path of length $n\geq0$ beginning at $s$ and ending at $s'$ determines a sequence of $n+1$ states
$$s,\  s_{\sigma_0},\ (s_{\sigma_0})_{\sigma_{1}},\  \dots,\  (((s_{\sigma_0})_{\sigma_{1}})_{\dots})_{\sigma_{n-1}}=s'$$
labeled by a suitable word $a_0 a_1 \cdots a_{n-1}$.
In the case $n=0$ we have that $s=s'$.

\begin{definition}[Strong connectivity]
An unrestricted Rabin automaton $\A =(S,\Sigma,A,\T)$ is said to be \emph{strongly connected} if for each pair of states
$s, s' \in S$, there is a bundle path beginning at $s$ and ending at $s'$.
\end{definition}

Note that a strongly connected Rabin automaton is also essential.

\begin{definition}[Irreducible tree shift]
A tree shift $X\subset A^{\Sigma^*}$ is \emph{irreducible} if for each pair of blocks $p \in X_n$ and $q\in X_m$, $n,m \geq 0$, there exists a configuration $f \in X$ satisfying the following properties:
$f\vert_{\Delta_n} = p$ and for each $w\in \Sigma^n$ there exists $v = v(w) \in \Sigma^*$ such that $f\vert_{wv\Delta_m} = wv \cdot q$ (see Definition~\ref{d:def w cdot p}). This situation is illustrated in Figure~\ref{FIGirreducibility}.
\end{definition}

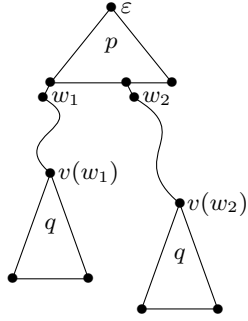
\begin{figure}[h!]
\scriptsize
\centering
%%%%%%%%%%%%%%%%%%%%%%%%%%%%%%%%%%%%%%%%%%%%%%%%%%%%%%%%%%%%%%%%%%%%%%%
\begin{picture}(0,40)(0,-20)
\gasset{Nw=0,Nh=0,Nframe=n,ATnb=0,AHnb=0}
\node(E)(0,20){$\bullet$}
\node(Enome)(2,20){$\varepsilon$}
\node(P)(0,15){$p$}
\node(0)(-8,10){$\bullet$}
\node(i)(2,10){$\bullet$}
\node(1)(8,10){$\bullet$}
\drawedge(E,0){}
\drawedge(E,1){}
\drawedge(0,1){}

\node(W)(-9,8){$\bullet$}
\node(Wnome)(-6,8){$w_1$}
\drawedge(0,W){}

\node(W')(3,8){$\bullet$}
\node(W'nome)(6,8){$w_2$}
\drawedge(i,W'){}

\node(E')(-8,-2){$\bullet$}
\node(E'nome)(-3,-2){$v(w_1)$}
\node(Q)(-8,-9){$q$}
\node(0')(-13,-16){$\bullet$}
\node(1')(-3,-16){$\bullet$}
\drawedge(E',0'){}
\drawedge(E',1'){}
\drawedge(0',1'){}

\drawbpedge(W,-30,8,E',140,8){}

\node(E'')(9,-6){$\bullet$}
\node(E''nome)(14,-6){$v(w_2)$}
\node(Q)(9,-13){$q$}
\node(0'')(4,-20){$\bullet$}
\node(1'')(14,-20){$\bullet$}
\drawedge(E'',0''){}
\drawedge(E'',1''){}
\drawedge(0'',1''){}

\drawbpedge(W',-30,10,E'',140,10){}
\end{picture}
%%%%%%%%%%%%%%%%%%%%%%%%%%%%%%%%%%%%%%%%%%%%%%%%%%%%%%%%%%%%%%%%%%%%%%%
\caption{Irreducibility of a tree shift.}
\label{FIGirreducibility}
\end{figure}

\begin{theorem}
\label{t:connected->irreducible}
Let $\A$ be a strongly connected unrestricted Rabin automaton. Then $\Sh_\A$ is irreducible.
\end{theorem}
\proof
Let $\A =(S,\Sigma,A,\T)$. Let $p, q\in \B(\Sh_\A)$ be two blocks of $\Sh_\A$ with support $\Delta_n$ and $\Delta_m$, respectively. Let $w\in \Sigma^n$ be a word of length $n$. Suppose $p$ is accepted via $\alpha \colon  \Delta_{n+1} \to  S$ (in the sense of Proposition~\ref{p:acceptanceSUBTREE2}) and that $q$ is accepted via $\alpha' \colon  \Delta_m \to  S$ (in the sense of Definition~\ref{d:acceptanceSUBTREE}). Fix $s = \alpha(w)$ and $s' = \alpha'(\varepsilon)$. Since $\A$ is strongly connected, there is a bundle path $\pi$ of length $h\geq0$ beginning at $s$ and ending at $s'$ yielding a sequence of $h+1$ states $$s,\  s_{\sigma_0},\ (s_{\sigma_0})_{\sigma_{1}},\  \dots,\  (((s_{\sigma_0})_{\sigma_{1}})_{\dots})_{\sigma_{h-1}}=s'$$
labeled by $$\lambda(\pi) = a_0 a_1 \cdots  a_{h-1}.$$
Notice that $\pi = \pi(w)$ and, in particular, $\sigma_i = \sigma_i(w)$ for all $i=0, 1, \dots, h-1$.

Define $v = v(w) = \sigma_0\sigma_1\cdots\sigma_{h-1} \in \Sigma^{h}$ (this corresponds to $v = \varepsilon$ whenever $h=0$). The set $$T = \Delta_{n+1} \cup \left(\bigcup_{w\in \Sigma^n}\{w\sigma_0, w\sigma_0\sigma_1, \dots, w\sigma_0\sigma_1\dots\sigma_{h-2}\} \right) \cup \left(\bigcup_{w\in \Sigma^n}wv\Delta_m\right)
$$
is a subtree of $\Sigma^*$.

The map $\bar \alpha \colon  T \to  S$ defined by
$$\bar \alpha(u) =
\begin{cases}
\alpha(u) & \textup{if } u \in \Delta_{n+1}\\
(((s_{\sigma_0})_{\sigma_{1}})_{\dots})_{\sigma_i} &\textup{if } u=w\sigma_0\sigma_1\cdots\sigma_{i} \textup{ for some } w\in\Sigma^n \textup{ and }i\in\{0, \dots, h-2\}\\
\alpha'(v') & \textup{if } u \in wv\Delta_m \textup{ for some } w\in\Sigma^n \textup{ and } u = wvv',
\end{cases}$$
accepts the pattern $\bar p \in A^T$ defined as follows
$$\bar p(u) =
\begin{cases}
p(u) & \textup{if } u \in \Delta_n\\
a_0 & \textup{if } u \in \Sigma^n\\
a_{i+1} &\textup{if } u=w\sigma_0\sigma_1\cdots\sigma_{i} \textup{ for some } w\in\Sigma^n \textup{ and }i\in\{0, \dots, h-2\}\\
q(v') & \textup{if } u \in wv\Delta_m \textup{ for some } w\in\Sigma^n \textup{ and } u = wvv'.
\end{cases}$$
A configuration  $f \in \Sh_\A$ extending $\bar p$ exists by Theorem~\ref{t:subtree}. Thus the shift $\Sh_\A$ is irreducible.
\endproof
Also the converse of the above theorem holds: an irreducible sofic tree shift admits a suitable strongly connected presentation. This will be proved in the next section, by using the notion of co-determinism.

\section{Deterministic and co-deterministic presentations}
\begin{definition}[Determinism and co-determinism]\label{d:CODET}
An unrestricted Rabin automaton $\A =(S,\Sigma,A,\T)$ is said to be \emph{deterministic} if, for each state $s\in S$, the transition bundles starting at $s$ carry different labels. Analogously, $\A$ is said to be \emph{co-deterministic} if, for each sequence ${\bf s} \in S^\Sigma$, the transition bundles terminating at ${\bf s}$ (if there are any) carry different labels.
\end{definition}

\begin{example}
Let $X \subset A^{\Sigma^*}$ be a tree shift of finite type. Then the unrestricted Rabin automaton accepting $X$ illustrated in Remark~\ref{r:transducer} is co-deterministic. If $n$ is the memory of $X$, this automaton is also \emph{$n$-local}, in the sense that if a block of size $n$ is accepted by two maps according to Definition~\ref{d:acceptanceSUBTREE}, then they coincide on $\varepsilon$.
In fact, Proposition~\ref{p:A(tau,M,X) accepts tau(X)} proves that a tree shift of finite type is always accepted by a co-deterministic and local unrestricted Rabin automaton. This implication and its converse have also been proved by Aubrun~\cite{Aubrun11}.
\end{example}

We now prove that for each unrestricted Rabin automaton there exists a co-deterministic unrestricted Rabin automaton accepting the same shift. To do this, we generalize the well-known ``subset construction'' to our class of automata.

\begin{theorem}[Subset construction]\label{t:subsetCONST}
Let $\A$ be an unrestricted Rabin automaton. There exists a co-deterministic and essential unrestricted Rabin automaton $\A_{\rm cod}$ such that $\Sh_\A = \Sh_{\A_{\rm cod}}$.
\end{theorem}
\proof
Let $\A = (S,\Sigma,A,\T)$ be an essential Rabin automaton (see Remark~\ref{r:essential}). We define $\A_{\rm cod} = ({\mathcal P}^*(S),\Sigma,A,\T_{\rm cod})$ as follows. The state set of $\A_{\rm cod}$ is the set ${\mathcal P}^*(S)$ of all nonempty subsets of the state set $S$ of $\A$. If $(\mathcal{S}_\sigma)_{\sigma \in \Sigma} \in ({\mathcal P}^*(S))^\Sigma$ and $a \in A$, let $\mathcal{S}$ denote the source state set of all the transition bundles in $\T$ labeled by $a$ and terminating at some $(s_\sigma)_{\sigma\in\Sigma} \in S^\Sigma$ with $s_\sigma \in \mathcal{S}_\sigma$ for each $\sigma \in \Sigma$.
If $\mathcal{S}$ is nonempty, we impose that the transition bundle $(\mathcal{S};a; (\mathcal{S}_\sigma)_{\sigma \in \Sigma})$ belongs to $\T_{\rm cod}$.

It is clear by construction that the automaton $\A_{\rm cod}$ is co-deterministic.
We have to prove that $\Sh_\A = \Sh_{\A_{\rm cod}}$. By Remark~\ref{r:B(X) determines X}, it suffices to show that $\B(\Sh_\A) = \B(\Sh_{\A_{\rm cod}})$.

Suppose that $p \in \B(\Sh_\A)$ is a block, say $p \colon \Delta_n \to A$. Then there exists a map $\alpha \colon  \Delta_{n+1} \to  S$ such that
$(\alpha(w);p(w);(\alpha(w\sigma))_{\sigma \in \Sigma}) \in \T$ for each $w \in \Delta_n$.
We inductively define $\alpha_{\rm cod} \colon \Delta_{n+1} \to  {\mathcal P}^*(S)$ on $\Sigma^{n-i}$ for $i = 0, \dots, n$. Moreover, we want that $\alpha(w) \in \alpha_{\rm cod}(w)$ for each $w \in \Delta_{n+1}$. For $i = 0$ and $w \in \Sigma^n$ we
define $\alpha_{\rm cod}(w) = \{\alpha(w)\}$. If $\alpha_{\rm cod}$ has been defined on $\Sigma^{n-i}$ and $n-i-1 \geq 0$, consider $w \in \Sigma^{n-i-1}$. We have that $(\alpha(w);p(w);(\alpha(w\sigma))_{\sigma \in \Sigma}) \in \T$. By our induction hypothesis, we have that $\alpha(w\sigma) \in \alpha_{\rm cod}(w\sigma)$ for each $\sigma\in \Sigma$. Hence there is $\mathcal{S}\in {\mathcal P}^*(S)$ such that $\alpha(w) \in \mathcal{S}$ and $(\mathcal{S};p(w);(\alpha_{\rm cod}(w\sigma))_{\sigma \in \Sigma})\in \T_{\rm cod}$. Define $\alpha_{\rm cod}(w) = \mathcal{S}$. It follows that the map $\alpha_{\rm cod}$ accepts~$p$.
Notice that $\A_{\rm cod}$ need not to be essential and then $p$ need not to be extensible to a configuration in $\Sh_{\A_{\rm cod}}$. In order to conclude that $p \in \B(\Sh_{\A_{\rm cod}})$, we need to observe that the map $\alpha$ could have been extended to a map $\alpha' \colon \Delta_{m+1} \to A$ with $m \geq n$ in a way such that each state in $\alpha'(\Sigma^m)$ is contained in a suitable $\mathcal{S} \subset S$ which is the source state (i.e., the image of $\varepsilon$) of an homomorphism determining a configuration accepted by $\A_{\rm cod}$. With this, we have that $\B(\Sh_\A) \subset \B(\Sh_{\A_{\rm cod}})$.

Conversely, suppose that $p \in \B(\Sh_{\A_{\rm cod}})$, with $p \colon  \Delta_n \to A$. Then there exists a map $\alpha_{\rm cod} \colon  \Delta_{n+1} \to {\mathcal P}^*(S)$ such that
$(\alpha_{\rm cod}(w);p(w);(\alpha_{\rm cod}(w\sigma))_{\sigma \in \Sigma}) \in \T_{\rm cod}$ for each $w \in \Delta_n$.
We inductively define $\alpha \colon  \Delta_{n+1} \to  S$ on $\Sigma^i$, for $i = 0, \dots, n$. Moreover, we want that $\alpha(w) \in \alpha_{\rm cod}(w)$ for each $w \in \Delta_{n+1}$. For $i = 0$, we define $\alpha(\varepsilon)$ as an arbitrary element of $\alpha_{\rm cod}(\varepsilon)$.
If $\alpha$ has been defined on $\Sigma^i$ and $i+1 \leq n$, consider $w \in \Sigma^i$.
Since $\alpha(w) \in \alpha_{\rm cod}(w)$  and $(\alpha_{\rm cod}(w);p(w);(\alpha_{\rm cod}(w\sigma))_{\sigma \in \Sigma}) \in \T_{\rm cod}$, there exists a sequence $(s_\sigma)_{\sigma\in\Sigma} \in S^\Sigma$ such that $(\alpha(w);p(w);(s_\sigma)_{\sigma\in\Sigma})\in \T$ with $s_\sigma \in \alpha_{\rm cod}(w\sigma)$ for each $\sigma \in \Sigma$. We define $\alpha(w\sigma) = s_\sigma$.
It follows that the map $\alpha$ accepts $p$, so that $p \in \B(\Sh_{\A})$.  This shows that $\B(\Sh_{\A_{\rm cod}}) \subset \B(\Sh_\A)$.
\endproof
The statement of the above theorem fails to hold, in general, for deterministic unrestricted Rabin automata, as shown in the following counterexample.
%This equivalence between general unrestricted Rabin automata and the co-deterministic ones is used to show an algorithm to establish whether a cellular automaton defined on $A^{\Sigma^*}$ is surjective or not.
\begin{example}[A shift of finite type not admitting a deterministic presentation]\label{e:countnondet}
Consider the tree shift $X$ presented in Example~\ref{e:monochromaticchildren}. A nondeterministic presentation of $X$ is given in Example~\ref{e:AUTmonochromaticchildren}.
Suppose that $X$ admits a deterministic presentation $\A = (S,\Sigma,A,\T)$. First observe that, in this case, each \emph{accessible} state (i.e., reachable by a transition bundle), admits exactly one transition bundle starting at it. Thus for every accessible state $s \in S$ there exists exactly one configuration $f_s \in X$ accepted by a homomorphism $\alpha_s \colon \Sigma^* \to S$ \emph{beginning} at $s$, that is, such that $\alpha_s(\varepsilon) = s$.
This implies that any state determines at most $|A|$ configurations (indeed, for a state $s$ that is not accessible, there are at most $|A|$ bundles that begin at $s$ and all of these bundles terminate in accessible states). Therefore $\A$ accepts only finitely many different configurations, which contradicts the fact that $X$ is infinite.
\end{example}

\subsection{Minimal presentations}
Let $\Sigma$ be a nonempty finite set and let $A$ be a finite alphabet.

\begin{definition}[Minimal presentation]
Let $X \subset A^{\Sigma^*}$ be a sofic tree shift.
A \emph{minimal co-deterministic presentation} of $X$ is a co-deterministic presentation of $X$ having fewest states among all co-deterministic presentations of $X$.
\end{definition}

\begin{theorem}
\label{t:irreducible->connected}
Let $X \subset A^{\Sigma^*}$ be an irreducible sofic tree shift and let $\A$ be a minimal co-deterministic presentation of $X$. Then $\A$ is strongly connected.
\end{theorem}
\proof
Let $\A =(S,\Sigma,A,\T)$. We first prove that for each state $s \in S$, there is a block $p_s\in \B(X)$ such that every map $\alpha$ accepting $p_s$ (in the sense of Proposition~\ref{p:acceptanceSUBTREE2}) involves $s$, that is, $s \in \im(\alpha)$. Suppose that this is not the case for some $s\in S$ and consider the automaton
$\A' = (S\setminus\{s\},\Sigma, A, \T')$ where $\T'$ consists of the transition bundles $t \in \T$ not involving $s$ (that is, ${\bf i}(t) \neq s \neq {\bf t}_\sigma(t)$ for all $\sigma \in \Sigma$). We have that $\B(\Sh_\A) = \B(\Sh_{\A'})$ and therefore $\Sh_\A = \Sh_{\A'}$ by Remark~\ref{r:B(X) determines X}. Thus $\A'$ is a co-deterministic presentation of $X$ with fewer states than $\A$, contradicting the minimality of $\A$.

Let now $s, s' \in S$ be two states. Consider the two blocks $p_s$ and $p_{s'}$ as above, say with support $\Delta_n$ and $\Delta_m$, respectively. Since $X$ is irreducible, there exists a configuration $f \in X$ satisfying the following conditions: $f$ coincides with $p_s$ on $\Delta_n$, and for each $w\in \Sigma^n$ there exists $v = v(w) \in \Sigma^*$ such that $f$ coincides with $wv \cdot p_{s'}$ on $wv\Delta_m$.
Let $\alpha$ be a homomorphism accepting $f$. Hence there is $u\in \Delta_{n+1}$ such that $\alpha(u) = s$ and for a fixed $w\in \Sigma^n$ in the descendance of $u$ there exist $v = v(w)$ and $u'\in \Delta_m$ such that $\alpha(wvu') = s'$. The image under $\alpha$ of the path in $\Sigma^*$ from $u$ to $wvu'$ determines a bundle path from $s$ to $s'$. Thus $\A$ is strongly connected.
\endproof

\section{Finite-tree automata}
\begin{definition}\label{d:finite tree acceptance}
A \emph{finite-tree automaton} is an unrestricted Rabin automaton $\A = (S,\Sigma,$ $A,\T)$ for which a subset $\mathcal{I} \subset S$ of \emph{initial states} and a state $F \in S$, called \emph{final state}, are specified. We shall denote it by $\A(\mathcal{I},F)$.

Let $T\subset \Sigma^*$ be a finite full subtree. We say that a full-tree-pattern $p \in A^T$ is \emph{accepted by $\A(\mathcal{I}, F)$} if there exists a map $\alpha \colon  T^+ \to  S$ such that
\begin{enumerate}[(1)]
\item $p$ is accepted by $\A$ via $\alpha$ (see Proposition~\ref{p:acceptanceSUBTREE2});
\item $\alpha(\varepsilon) \in \mathcal{I}$;
\item $\alpha(w) = F$ if $w \in T^+ \setminus T$.
\end{enumerate}
We denote by $\TT(\A(\mathcal{I}, F))$ the set of all full-tree-patterns accepted by $\A(\mathcal{I}, F)$.
A set of full-tree-patterns is called \emph{recognizable} if it is of the form $\TT(\A(\mathcal{I}, F))$, for some finite-tree automaton $\A(\mathcal{I}, F)$.
A finite-tree automaton $\A(\mathcal{I}, F)$ is \emph{co-deterministic} if its underlying unrestricted Rabin automaton $\A$ is co-deterministic.
\end{definition}
\begin{remark}
As explained in Remark~\ref{r:essential}, we could only consider essential unrestricted Rabin automata.
As finite-tree automata are concerned, we relax this assumption: each non final state is the source of some transition bundle, but no condition is required for the final state. Thus in the finite-tree automata we consider, one may have no transition bundles starting from the final state.
\end{remark}
\begin{definition}
An unrestricted Rabin automaton $\A = (S,\Sigma,A,\T)$ is called \emph{co-complete} if for each ${\bf s} \in S^\Sigma$ and $a \in A$, there exists a transition bundle in $\T$ labeled by $a$ and terminating at~${\bf s}$. A finite-tree automaton $\A(\mathcal{I}, F)$ is \emph{co-complete} if its underlying unrestricted Rabin automaton $\A$ is co-complete.
\end{definition}
\begin{remark}\label{r:cocomplete+codet}
Let $\A = (S,\Sigma,A,\T)$ be a co-complete and co-deterministic unrestricted Rabin automaton and let $T\subset \Sigma^*$ be a finite full subtree. For each $F \in S$ and $p\in A^T$, there always exists a map $\alpha_p \colon  T^+ \to  S$ satisfying conditions $(1)$ and $(3)$ in Definition~\ref{d:finite tree acceptance}. Moreover, since $\A$ is co-deterministic, such a map is unique.
\end{remark}
A slight adaptation in the proof of Theorem~\ref{t:subsetCONST} leads to the proof of Theorem~\ref{t:subsetCONST2}.

\proof[Proof of Theorem~\ref{t:subsetCONST2}]
Let $\A = (S,\Sigma,A,\T)$ be an essential Rabin automaton. The unrestricted Rabin automaton $\A_{\rm cod} = ({\mathcal P}^*(S),\Sigma,A,\T_{\rm cod})$ is defined as in the proof of Theorem~\ref{t:subsetCONST}. Set $\mathcal{I} = {\mathcal P}^*(S)$ and $F = S$.

We have to prove that $\TT(\Sh_\A) = \TT(\A_{\rm cod}(\mathcal{I},F))$.
%Follow the proof of this of Theorem~\ref{t:subsetCONST} and replace $\Delta_n$ by a finite full subtree $T$.

Suppose that $p \in \TT(\Sh_\A)$ is a full-tree-pattern with support a finite full subtree $T$. Then there exists a map $\alpha \colon  T^+ \to  S$ such that
$(\alpha(w);p(w);(\alpha(w\sigma))_{\sigma \in \Sigma}) \in \T$ for each $w \in T$.

Set $\alpha_{\rm cod}(w) = S$ for each $w \in T^+\setminus T$. As the proof of Theorem~\ref{t:subsetCONST}, a map
$\alpha_{\rm cod} \colon T^+ \to  {\mathcal P}^*(S)$ can be inductively defined on $\Sigma^{n-i} \cap T$ for $i = 1, \dots, n$, where $n$ is the height of $T$. This way,
the map $\alpha_{\rm cod}$ accepts $p$ (in the sense of Definition~\ref{d:finite tree acceptance}). Thus $p \in \TT(\A_{\rm cod}(\mathcal{I},F))$. This shows that $\TT(\Sh_\A) \subset \TT(\A_{\rm cod}(\mathcal{I},F))$.

Conversely, suppose that $p \in \TT(\A_{\rm cod}(\mathcal{I},F))$ is a full-tree-pattern with support a finite full subtree $T$. Then there exists a map $\alpha_{\rm cod} \colon  T^+ \to {\mathcal P}^*(S)$ such that
$(\alpha_{\rm cod}(w);p(w);$ $(\alpha_{\rm cod}(w\sigma))_{\sigma \in \Sigma}) \in \T_{\rm cod}$ for each $w \in T$.
Again, as in the proof of Theorem~\ref{t:subsetCONST}, we can inductively define $\alpha \colon  T^+ \to  S$ on $\Sigma^i \cap T^+$, for $i = 0, \dots, n$, where $n$ is the height of $T$. This way, the map $\alpha$ accepts $p$. By Corollary~\ref{c:subtree}, we have $p \in \TT(\Sh_{\A})$.  This shows that $\TT(\A_{\rm cod}(\mathcal{I},F)) \subset \TT(\Sh_\A)$.\endproof

As proved in the following theorem, the recognizable sets of full-tree-patterns form a class which is closed under complementation.

\begin{theorem}\label{t:complement}Let $\A(\mathcal{I}, F)$ be a co-deterministic finite-tree automaton. Then there exists a co-complete and co-deterministic finite-tree automaton $\A_\complement(\mathcal{I}_\complement, F_\complement)$ such that $$\TT(A^{\Sigma^*}) \setminus \TT(\A(\mathcal{I}, F)) = \TT(\A_\complement(\mathcal{I}_\complement, F_\complement)).$$
\end{theorem}
\proof
Let $\A = (S,\Sigma,A,\T)$ be the co-deterministic unrestricted Rabin automaton underlying $\A(\mathcal{I}, F)$. We define $\A_\complement = (S_\complement,\Sigma,A,\T_\complement)$ as follows. We construct $S_\complement$ by adding to $S$ a new state $K \notin S$, so that $S_\complement = S \cup \{K\}$.
The set of transition bundles $\T_\complement$ contains all the transition bundles in $\T$. Moreover, for a sequence ${\bf s} \in (S_\complement)^\Sigma$ and $a \in A$, if no transition bundle in $\T$ labeled by $a$ terminates at
${\bf s}$, then we add to $\T_\complement$ the transition bundle $(K;a;{\bf s})$.
%Finally, for each $a\in A$, we add to $\T_\complement$ a bundle loop $(K;a;(K)_{\sigma \in \Sigma})$.
With this, we have that $\A_\complement$ is  co-deterministic and co-complete.

By setting $\mathcal{I}_\complement = \left(S \setminus \mathcal{I}\right) \cup \{K\}$ and $F_\complement = F$, we want to prove that a full-tree-pattern $p \in A^T$ is accepted by $\A(\mathcal{I}, F)$ if and only if it is not accepted by $\A_\complement(\mathcal{I}_\complement, F_\complement)$.
By Re\-mark~\ref{r:cocomplete+codet}, given $p\in A^T$ there exists exactly one map $\alpha_p \colon  T^+ \to  S_\complement$ satisfying conditions $(1)$ and $(3)$ in Definition~\ref{d:finite tree acceptance}.

Suppose that $p$ is accepted by $\A(\mathcal{I}, F)$ via $\alpha \colon  T^+ \to  S$ as in Definition~\ref{d:finite tree acceptance}. Recall that $S \subset S_\complement$ and $F = F_\complement$. By uniqueness, it follows that $\alpha = \alpha_p$. As $\alpha_p(\varepsilon) = \alpha(\varepsilon) \in \mathcal{I}$, we have $\alpha_p(\varepsilon) \notin \mathcal{I}_\complement$ and then
$p \notin \TT(\A_\complement(\mathcal{I}_\complement, F_\complement))$.

Suppose now that $p$ is not accepted by $\A(\mathcal{I}, F)$. The map $\alpha_p$ cannot map $\varepsilon$ to a state $s \in \mathcal{I}$ unless $K \in \alpha_p(T^+)$.
But this last condition is impossible because
there is no transition bundle $t \in \T_\complement$ with ${\bf i}(t) \neq K$ and ${\bf t}_\sigma(t) = K$ for some
$\sigma \in \Sigma$.  Hence $\alpha_p(\varepsilon) \in \mathcal{I}_\complement$ and $p$ is accepted by $\A_\complement(\mathcal{I}_\complement, F_\complement)$.
\endproof
We are now in position to prove Theorem~\ref{t:complement2} from the Introduction.

\proof[Proof of Theorem~\ref{t:complement2}]
Let $\A = (S,\Sigma,A,\T)$ be an unrestricted Rabin automaton. Let us show that there exists a co-complete and co-determi\-ni\-stic finite-tree automaton $\A_\complement(I,F)$ with a single initial state, such that $\TT(A^{\Sigma^*}) \setminus \TT(\Sh_\A) = \TT(\A_\complement(I,F))$.
By virtue of Theorem~\ref{t:subsetCONST2}, we can find a co-deterministic unrestricted Rabin automaton $\A_{\rm cod} = ({\mathcal P}^*(S),\Sigma,A,\T_{\rm cod})$ such that
$\TT(\Sh_\A) = \TT(\A_{\rm cod}({\mathcal P}^*(S), S))$.
By Theorem~\ref{t:complement}, there exists a co-complete and co-deterministic unrestricted Rabin automaton
$\A_\complement = ({\mathcal P}^*(S)\cup\{K\},\Sigma,A,\T_\complement)$ such that $\TT(A^{\Sigma^*}) \setminus \TT(\A_{\rm cod}({\mathcal P}^*(S), S)) = \TT(\A_\complement(K, S))$.
\endproof
\begin{corollary}\label{c:full}
Let $\A$ be an unrestricted Rabin automaton. Let $\A_\complement(I,F)$ be as in Theorem~\ref{t:complement2}. Then $\Sh_\A = A^{\Sigma^*}$ if and only if $\TT(\A_\complement(I, F)) = \varnothing$.
\end{corollary}
\proof
By Theorem~\ref{t:complement2}, we have that $\TT(\Sh_\A) = \TT(A^{\Sigma^*})$ if and only if $\TT(\A_\complement(I, F))=\varnothing$.
By Remark~\ref{r:T(X) determines X}, we have that $\Sh_\A = A^{\Sigma^*}$ if and only if
$\TT(\Sh_\A) = \TT(A^{\Sigma^*})$.
\endproof

\section{Decision problems}\label{s:Decision problems}

\subsection{The emptiness problem for finite-tree automata}\label{ss:emptiness}
The emptiness problem for an unrestricted Rabin automaton is trivial (every nonempty essential automaton accepts at least a configuration), but this argument does not apply to the case of finite-tree automata. In this section we present an effective procedure to establish the emptiness of a recognizable set of full-tree-patterns.

\proof[Proof of Theorem~\ref{t:emptiness}]
Let $\A(\mathcal{I}, F)$ be a finite-tree automaton. Let us show that there is an algorithm which establishes whether $\TT(\A(\mathcal{I}, F)) = \varnothing$ or not.
We claim that $\TT(\A(\mathcal{I}, F))$ is nonempty if and only if it contains a pattern of height $\leq \vert S \vert$, where $S$ is the state set of $\A$. Indeed, suppose that there exists a full-tree-pattern $p$ defined on a finite full subtree $T$ which is accepted by $\A(\mathcal{I}, F)$ (see Definition~\ref{d:finite tree acceptance}) via a map $\alpha \colon  T^+ \to  S$. If a branch contains more than $\vert S \vert$ vertices, this means that a state in $S$ appears at least twice in the $\alpha$-image of this branch. More precisely, there exist $w, w' \in T$ with $w' \in w\Sigma^* \setminus \{w\}$ and such that $\alpha(w) = \alpha(w')$. One can always consider $w$ and $w'$ as the shortest and the longest word in the branch having this property, respectively. Define $T' = (T \setminus w\Sigma^*)\cup \{wv : v \in \Sigma^* \textup{ and } w'v \in T\}$. In other words, $T'$ is obtained by replacing the descendants of $w$ with the descendants of $w'$. Obviously $T'$ is still a finite full subtree. Consider the full-tree-pattern $p' \colon  T' \to  A$ defined by $p'\vert_{T \setminus w\Sigma^*} = p\vert_{T \setminus w\Sigma^*}$ and $p'(wv) = p(w'v)$ for all $v \in \Sigma^*$ such that $w'v\in T$. As it can be easily seen, $p'$ is accepted by $\A(\mathcal{I}, F)$ via a suitable modification of $\alpha$. Thus, $\TT(\A(\mathcal{I}, F))$ accepts a full-tree-pattern defined on a finite full subtree $T'$ obtained by a nontrivial shortening of the appointed branch.
Up to recursively applying this process, we get a full-tree-pattern in which every branch has at most $\vert S \vert$ vertices of $\Sigma^*$ involved.
One then concludes since there are finitely many full-tree-patterns of height $\leq \vert S \vert$ and one can effectively check whether they are accepted by $\A(\mathcal{I}, F)$ or not.
\endproof

\begin{remark}
Following the above proof, in order to solve the emptiness problem we have, in principle, to check all possible maps $\alpha \colon \Delta_{\vert S \vert +1} \to  S$. Thus, the previous algorithm has exponential complexity in the size of $S$.
\end{remark}

\subsection{An algorithm establishing whether two sofic tree shifts coincide or not}
%Let $\A_1$ and $\A_2$ be two unrestricted Rabin automata. We want to establish whether or not $\Sh_{\A_1}=\Sh_{\A_2}$.
%Note that, by Corollary~\ref{T(X) determines X}, it suffices to establish whether $\TT(\Sh_{\A_1}) \setminus \TT(\Sh_{\A_2}) = \varnothing = \TT(\Sh_{\A_2}) \setminus \TT(\Sh_{\A_1})$.
Before presenting our algorithm, we need to define
the join of two unrestricted Rabin automata.

\begin{definition}\label{d:join}
The \emph{join} of two unrestricted Rabin automata $\A_1 = (S_1,\Sigma,A,\T_1)$ and $\A_2 = (S_2,\Sigma,A,\T_2)$ is the unrestricted Rabin automaton $\A_1 * \A_2 = (S_1 \times S_2,\Sigma,A,\T_\times)$ where
$$\left((s_1,s_2);a;(s'_\sigma,s''_\sigma)_{\sigma \in \Sigma}\right) \in \T_\times \Longleftrightarrow \left(s_1;a;(s'_\sigma)_{\sigma \in \Sigma}\right) \in \T_1 \textup{ and }
\left(s_2;a;(s''_\sigma)_{\sigma \in \Sigma}\right) \in \T_2.$$
\end{definition}

Notice that $\Sh_{\A_1 * \A_2} = \Sh_{\A_1} \cap \Sh_{\A_2}$. Moreover, $\A_1 * \A_2$ is co-complete (resp. co-deterministic), if $\A_1$ and $\A_2$ are co-complete (resp. co-deterministic).

%Returning back to our problem, given the two unrestricted Rabin automata $\A_1$ and $\A_2$, we first construct the co-complete and co-deterministic unrestricted Rabin automata $\A_1'(I_1,F_1)$ and $\A_2'(I_2,F_2)$ as in Theorem~\ref{complement2}, associated with $\A_1$ and $\A_2$, respectively.
%Consider the unrestricted Rabin automaton $(\A_1'*\A_2')(\mathcal{I}_1,F)$, where $\mathcal{I}_1 = (S_1\setminus\{I_1\}) \times \{I_2\}$ and $F = (F_1,F_2)$.
%For each full-tree-pattern $p \in \TT(A^{\Sigma^*})$ we have
%$$p \in \TT((\A_1'*\A_2')(\mathcal{I}_1,F)) \Longleftrightarrow
%p \notin \TT(\A_1'(I_1,F_1)) \textup{ and } p \in \TT(\A_2'(I_2,F_2)),$$ that is,
%$$p \in \TT((\A_1'*\A_2')(\mathcal{I}_1,F)) \Longleftrightarrow
%p \in \TT(\Sh_{\A_1}) \setminus \TT(\Sh_{\A_2}).$$
%Analogously, by defining $\mathcal{I}_2 = \{I_1\}\times (S_1\setminus\{I_2\})$ one has
%$$p \in \TT((\A_1'*\A_2')(\mathcal{I}_2,F)) \Longleftrightarrow
%p \in \TT(\Sh_{\A_2}) \setminus \TT(\Sh_{\A_1}).$$
%
%Thus $\Sh_{\A_1}=\Sh_{\A_2}$ if and only if $\TT((\A_1'*\A_2')(\mathcal{I}_1,F)) \bigcup \TT((\A_1'*\A_2')(\mathcal{I}_2,F)) = \varnothing$. An effective procedure to establish this latter equality is provided by Corollary~\ref{emptiness}.
%
We are now in position to prove Theorem~\ref{t:equality} from the Introduction.

\proof[Proof of Theorem~\ref{t:equality}]
Let $\A_1$ and $\A_2$ be two unrestricted Rabin automata. Note that, by Remark~\ref{r:T(X) determines X}, it suffices to establish whether
\begin{equation}
\label{whether-or-not}
\TT(\Sh_{\A_1}) \setminus \TT(\Sh_{\A_2}) = \varnothing = \TT(\Sh_{\A_2}) \setminus \TT(\Sh_{\A_1})
\end{equation}
or not.

First construct the co-complete and co-deterministic finite-tree automata $\A_1'(I_1,F_1)$ and $\A_2'(I_2,F_2)$ as in Theorem~\ref{t:complement2}, associated with $\A_1$ and $\A_2$, respectively. Hence $\TT(A^{\Sigma^*}) \setminus \TT(\Sh_{\A_i}) = \TT(\A_i'(I_i,F_i))$ for $i=1,2$. Moreover, we denote by $S_i$ the state set of $\A_i'$ for $i=1,2$.
Consider the finite-tree automaton $(\A_1'*\A_2')(\mathcal{I}_1,F)$, where $\mathcal{I}_1 = (S_1\setminus\{I_1\}) \times \{I_2\}$ and $F = (F_1,F_2)$.
For each full-tree-pattern $p \in \TT(A^{\Sigma^*})$ we have
$$p \in \TT((\A_1'*\A_2')(\mathcal{I}_1,F)) \Longleftrightarrow
p \notin \TT(\A_1'(I_1,F_1)) \textup{ and } p \in \TT(\A_2'(I_2,F_2)),$$ that is,
$\TT((\A_1'*\A_2')(\mathcal{I}_1,F)) = \TT(\Sh_{\A_1}) \setminus \TT(\Sh_{\A_2})$.
Analogously, by defining $\mathcal{I}_2 = \{I_1\}\times (S_2\setminus\{I_2\})$ one has
$\TT((\A_1'*\A_2')(\mathcal{I}_2,F)) = \TT(\Sh_{\A_2}) \setminus \TT(\Sh_{\A_1})$.

Thus \eqref{whether-or-not} holds if and only if $\TT((\A_1'*\A_2')(\mathcal{I}_1,F)) \bigcup \TT((\A_1'*\A_2')(\mathcal{I}_2,F)) = \varnothing$.
An effective procedure to establish this latter equality is then provided by Theorem~\ref{t:emptiness} proved in Section~\ref{ss:emptiness}.
\endproof

\begin{remark}
The above algorithm has exponential complexity in the maximal size of the state sets of the unrestricted Rabin automata. A different procedure can be applied to the class of \emph{strongly irreducible} unrestricted Rabin automata by using a minimization process. Actually, in~\cite{AubrunBeal10} it is shown that there exists a canonical minimal co-deterministic presentation of an irreducible sofic tree shift. Thus another possible decision algorithm consists in computing the minimal presentations of the two shifts and checking whether they coincide or not. In this case, Theorem~\ref{t:subsetCONST} is needed while the procedure for the emptiness problem is not required. Hence this algorithm has in general an exponential complexity. The complexity can be reduced to be polynomial by only considering the class of co-deterministic strongly irreducible tree shifts.
\end{remark}

\subsection{An algorithm establishing whether a cellular automaton is surjective or not}\label{ss:surjectivity}
We are now in position to prove Theorem~\ref{t:surjectivity} in the Introduction. Observe first
that giving a sofic tree shift $X \subset A^{\Sigma^*}$ corresponds, equivalently, to giving a tree shift of finite type $Z \subset C^{\Sigma^*}$ (for a suitable finite alphabet $C$) and a surjective cellular automaton $\tau' \colon Z \to X$, or an unrestricted Rabin automaton $\A$ such that $X = \Sh_\A$. Propositions~\ref{p:XA sofic} and~\ref{p:A(tau,M,X) accepts tau(X)} provide an effective procedure to switch from one representation to the other.

%\emph{Proof of Theorem~\ref{t:E}]
%Let $A$ and $B$ be two finite alphabets. Let $X \subset A^{\Sigma^*}$ and $Y \subset B^{\Sigma^*}$ be two sofic shifts and $\tau \colon X \to Y$ a CA.
%Let us show that it is decidable whether $\tau$ is surjective or not.
%Let $\A$ be an unrestricted Rabin automaton accepting $X$. Denote by $\mathfrak{A}$ the underlying bundle automaton and let $C$ denote its alphabet set. Then the subshift $Z:= \Sh_{\mathfrak{A}} \subset C^{\Sigma}$ is of finite type and the labeling map
%$\lambda \colon C \to A$ induces a CA (with memory set $M = \{\varepsilon\}$)
%$\tau' \colon C^{\Sigma^*} \to A^{\Sigma^*}$ such that $\tau'(Z) = X$.
%Now the CA $\tau \colon X \to Y$ is surjective if and only if
%the composite CA $\tau \circ \tau' \colon Z \to Y$ is surjective.
%Thus we can reduce to the case when $X$ is of finite type.
%Let $n \in \mathbb{N}$ be large enough so that the CA $\tau$ has memory set $\Delta_n$ and that $n-1$ is the memory of $X$. By Theorem~\ref{A(tau,M,X) accepts tau(X)}, the De Bruijn automaton $\A(\tau, \Delta_n, X)$ having state set $X_{n-1}$ is a presentation of $\tau(X)$. Then, it suffices to apply Theorem~\ref{t:D} to establish whether $Y=\Sh_{\A(\tau, \Delta_n, X)}$.
\proof[Proof of Theorem~\ref{t:surjectivity}]
Let $X \subset A^{\Sigma^*}$ and $Y \subset B^{\Sigma^*}$ be two sofic tree shifts and let $\tau \colon X \to Y$ be a CA.
We want to show that it is decidable whether $\tau$ is surjective or not.
Let $Z \subset C^{\Sigma^*}$ and $\tau' \colon Z \to X$ be as above.
Now the cellular automaton $\tau \colon X \to Y$ is surjective if and only if
the composite cellular automaton $\tau \circ \tau' \colon Z \to Y$ is surjective.
%Thus we can reduce to the case when $X$ is of finite type.
Let $n \in \mathbb{N}$ be large enough so that the cellular automaton $\tau \circ \tau'$ has memory set $\Delta_n$ and that $n-1$ is the memory of $Z$. By Proposition~\ref{p:A(tau,M,X) accepts tau(X)}, the unrestricted Rabin automaton $\A(\tau \circ \tau', \Delta_n, Z)$ having state set $Z_{n-1}$ is a presentation of $\tau(X)$. Then, it suffices to apply Theorem~\ref{t:equality} to establish whether $Y=\tau(X)$ or not.
\endproof
\begin{remark}
If in Theorem~\ref{t:surjectivity} one has $X = A^{\Sigma^*}$ and $Y = B^{\Sigma^*}$, then the algorithm becomes much simpler.
Indeed, once constructed the unrestricted Rabin automaton $\A = \A(\tau, \Delta_n, A^{\Sigma^*})$ accepting $\tau(A^{\Sigma^*})$ as in Proposition~\ref{p:A(tau,M,X) accepts tau(X)}, we have that (by virtue of Corollary~\ref{c:full}), the map $\tau$ is surjective if and only if
$\TT(\A_\complement(I,F)) = \varnothing$,
where $\A_\complement(I,F)$ is the co-complete and co-deterministic finite-tree automaton corresponding to $\A$ as in Theorem~\ref{t:complement2}. This latter equality is effectively verifiable by Theorem~\ref{t:emptiness} proved in Section~\ref{ss:emptiness}.
\end{remark}

\subsection{An algorithm establishing whether a cellular automaton is injective or not}
 We are now going to prove Theorem~\ref{t:injectivity} in the Introduction. Notice that we cannot apply the same argument used in Section~\ref{ss:surjectivity} for establishing whether a CA between sofic tree shifts is injective or not. We have to limit ourselves to the case of a CA defined on a tree shift of finite type.

\proof[Proof of Theorem~\ref{t:injectivity}]
Let $X \subset A^{\Sigma^*}$ be a tree shift of finite type and let $\tau \colon X \to B^{\Sigma^*}$ be a CA.
We want to show that it is decidable whether $\tau$ is injective or not.
Let $n \in \mathbb{N}$ be large enough so that the cellular automaton $\tau$ has memory set $\Delta_n$ and that $n-1$ is the memory of $X$. By Proposition~\ref{p:A(tau,M,X) accepts tau(X)}, the unrestricted Rabin automaton $\A = \A(\tau, \Delta_n, X)$ having state set $X_{n-1}$ is a presentation of $\tau(X)$.

As we have seen, the unrestricted Rabin automaton $\A*\A$ (see Definition~\ref{d:join}) is another presentation of
$\tau(X)$. Recall that the state set of $\A*\A$ is $S\times S$, where $S = X_{n-1}$ is the state set of $\A$.
A state $(s, s')$ of $\A*\A$ is \emph{diagonal} if $s
= s'$. Notice that the cellular automaton $\tau$ is noninjective if and only
if there exist a configuration $f\in\tau(X)$ and two different homomorphisms $\alpha, \alpha' \colon \mathcal{G}_f \to \A$.
This fact is equivalent to the existence
of a homomorphism $\alpha \colon \mathcal{G}_f \to \A*\A$
that involves a
nondiagonal state (i.e., such that the image $\alpha(\Sigma^*)$ contains a nondiagonal state).

Hence, starting from the Rabin automaton
$\A*\A$ we construct the essential Rabin automaton that accepts
the same sofic tree shift. It suffices to check, on this latter
automaton, if some nondiagonal state is left.
\endproof

\begin{remark}
Given a Rabin automaton $\A = (S,\Sigma,A,\T)$, the Rabin automaton $\A*\A$ has $\vert S \vert^2$ states and $O(\vert \T \vert^2)$ transition bundles.
The procedure to get the essential part of $\A*\A$ has then complexity $O(\vert S \vert^4\vert \T \vert^2)$ and it costs $O(\vert S \vert^2)$ to check the presence of a nondiagonal state.
Thus, the above algorithm has complexity $O(\vert S \vert^{2\vert\Sigma \vert+6})$ in the size of the state set of the unrestricted Rabin automaton $\A(\tau, \Delta_n, X)$ accepting $\tau(X)$.
\end{remark}

\section{Rabin automata}
\begin{definition}[Rabin automaton and acceptance]\label{d:rabin}
A \emph{Rabin automaton} is a 6-tuple $\mathcal{A}=(S,\Sigma,A,\mathcal{T},\mathcal{I},$ $\mathfrak{F})$ where $(S,\Sigma,A,\mathcal{T})$ is an unrestricted Rabin automaton, $\mathcal{I}\subset S$ is a set of \emph{initial states}, and $\mathfrak{F}\subset \mathcal{P}(S)$ is a family of \emph{accepting sets}.

A configuration $f \in A^{\Sigma^*}$ is \emph{accepted} by $\mathcal{A}$ if there exists a homomorphism $\alpha \colon \Sigma^* \to S$ such that
\begin{enumerate}[(1)]
\item $f$ is accepted by $(S,\Sigma,A,\mathcal{T})$ (see Remark~\ref{r:acceptance});
\item $\alpha(\varepsilon) \in \mathcal{I}$;
\item for every right infinite word $\bar w= \sigma_0\sigma_1 \cdots \in \Sigma^\mathbb{N}$, the set of states that appear infinitely many times along the image of $\bar w$ under $\alpha$ is in $\mathfrak{F}$, i.e.,
$$
 {\rm In}(\alpha\vert \bar w) := \{s \in S : s \text{ appears infinitely often in }
  \alpha(\varepsilon),\alpha(\sigma_0),\alpha(\sigma_0\sigma_1),\dots\ \} \in \mathfrak{F}.
$$
\end{enumerate}
\end{definition}

In this section we modify the definition of an essential automaton. Namely, a state $s \in S$ in a Rabin automaton $\mathcal{A}=(S,\Sigma,A,\mathcal{T},\mathcal{I},\mathfrak{F})$ is \emph{essential} if there exists at least one configuration $f$ and a homomorphism $\alpha: \mathcal{G}_f \to \A$ that accepts $f$ and \emph{uses} $s$ (i.e., $\alpha(w)=s$, for some $w \in \Sigma^*$). A Rabin automaton is \emph{essential} if all its states are essential.

We will prove Theorem~\ref{t:rabin+shift=sofic} from the Introduction through two lemmas, designed to indicate that accepting sets are not needed to recognize topologically closed sets, and initial sets are not needed to recognize shift-invariant sets.

\begin{lemma}\label{l:no-initial}
Let $X$ be a shift-invariant subset of $A^{\Sigma^*}$ that is accepted by the essential Rabin automaton $\A=(S,\Sigma,A,\T,\mathcal{I},\mathfrak{F})$. Then, $X$ is also accepted by the Rabin automaton $\A'=(S,\Sigma,A,\T,S,\mathfrak{F})$ in which all states in $S$ are initial.
\end{lemma}

\proof
Let $s \in S$ be a state which is not in $\mathcal{I}$. Let $f \in A^{\Sigma^*}$ be a configuration and let $\alpha \colon \Sigma^* \to S$ be a homomorphism accepting $f$ (as in Remark~\ref{r:acceptance}) such that $\alpha(\varepsilon)=s$ and ${\rm In}(\alpha\vert \bar w) \in \mathfrak{F}$ for every right infinite word $\bar w$.
All we need to show is that the configuration $f$ is in $X$.

Let $g \in X$ be a configuration accepted by $\A$ via $\beta \colon \Sigma^* \to S$ that uses $s$ and let $\beta(\varepsilon)=s_0 \in \mathcal{I}$ and $\alpha(w)=s$, for some $w \in \Sigma^*$ (note that $w \neq \varepsilon$). Define a configuration $f'$ that agrees with $g$ outside of $w\Sigma^*$ and such that $f'(wv)=f(v)$, for each $v \in \Sigma^*$ (that is, $(f')^w = f$). Define $\alpha'\colon \Sigma^* \to S$ that agrees with $\beta$ outside of $w\Sigma^*$ and such that $\alpha'(wv)=\alpha(v)$, for each $v \in \Sigma^*$. It is straightforward to check that $\alpha'$ accepts $f'$. Indeed, $\alpha'(\varepsilon) = \beta(\varepsilon)=s_0 \in \mathcal{I}$ and, for any right infinite word $w\bar w$ with prefix $w$, we have
${\rm In}(\alpha'\mid w\bar w) = {\rm In}(\alpha\mid \bar w) \in \mathfrak{F}$, since the acceptance condition is independent of finite prefixes.
On the other side, for any right infinite word $\bar w$ that does not have $w$ as a prefix, we have
${\rm In}(\alpha'\mid \bar w) = {\rm In}(\beta\mid w) \in \mathfrak{F}$.
Therefore $f' \in X$. Since $f = (f')^w$ and $X$ is shift-invariant, we conclude that $f \in X$.
\endproof

\begin{lemma}\label{l:no-acceptance}
Let $X$ be a topologically closed subset of $A^{\Sigma^*}$ that is accepted by the essential Rabin automaton $\A=(S,\Sigma,A,\T,\mathcal{I},\mathfrak{F})$. Then, $X$ is also accepted by the Rabin automaton $\A'=(S,\Sigma,A,\T,\mathcal{I},\mathcal{P}(S))$ in which all subsets of $S$ are accepting.
\end{lemma}

\proof
For $s \in S$, let $f \in X$ be a configuration accepted by $\mathcal{A}$ via $\alpha\colon \Sigma^* \to S$ that uses $s$ and let $\alpha(\varepsilon)=s_0 \in \mathcal{I}$ and $\alpha(w)=s$, for some $w \in \Sigma^*$. Define $f_s=f^w$ and $\alpha_s\colon \Sigma^* \to S$ by $\alpha_s(v) = \alpha(wv)$, for $v \in \Sigma^*$. Note that, $\alpha_s(\varepsilon)=s$ and, since the accepting condition is not affected by initial subwords, for every right infinite word $\bar w$ we have
${\rm In}(\alpha_s \mid \bar w) \in \mathfrak{F}$.

Let $g \in A^{\Sigma^*}$ be a configuration and let $\beta\colon  \Sigma^* \to S$ be a homomorphism accepting $g$ (as in Remark~\ref{r:acceptance}) such that $\beta(\varepsilon) \in \mathcal{I}$. We will show that the configuration $g$ is in $X$. In other words, we will show that the accepting condition is irrelevant. Indeed, for $n \geq 1$, define a configuration $g_n$ that agrees with $g$ on $\Delta_n$ and such that $g_n(wv)=f_{\beta(w)}(v)$, for every $w \in \Sigma^n$ and $v\in\Sigma^*$ (that is, the configuration $g_n$ behaves as $f_{\beta(w)}$ on $w\Sigma^*$). Define a map $\beta_n\colon  \Sigma^* \to S$
 that agrees with $\beta$ on $\Delta_n$ and such that $\beta_n(wv)=\alpha_{\beta(w)}(v)$, for every $w \in \Sigma^n$ and $v\in\Sigma^*$. It is straightforward to check that $\beta_n$ accepts $g_n$. Indeed, $\beta_n(\varepsilon) = \beta(\varepsilon) \in \mathcal{I}$ and, for any right infinite word $\bar w$ and any finite word $w \in \Sigma^n$, we have
${\rm In}(\beta_n\mid \bar w) = {\rm In}(\beta_n\mid w\bar w) = {\rm In}(\alpha_{\beta(w)}\mid \bar w) \in \mathfrak{F}$.
Therefore $g_n \in X$ for each $n \geq 1$. Since $\lim_{n \to \infty}g_n =g$ and $X$ is closed, we conclude that $g \in X$.
\endproof

\proof[Proof of Theorem~\ref{t:rabin+shift=sofic}]
Clearly, by Corollary~\ref{c:sofic iff accepted}, any sofic tree shift is recognized by a Rabin automaton.

Conversely, if $X$ is a shift recognized by a Rabin automaton, Lemma~\ref{l:no-initial} and Lemma~\ref{l:no-acceptance} imply that it is also recognized by an unrestricted Rabin automaton, and therefore, by Corollary~\ref{c:sofic iff accepted}, $X$ is a sofic tree shift.
\endproof

\bibliographystyle{siam}
\bibliography{biblio}

\end{document}